\documentclass[letterpaper, 11pt]{article}
\usepackage{latexsym}
\usepackage{setspace}

\usepackage[margin=1in]{geometry}

\usepackage{amsmath, amsthm, amssymb, microtype, enumerate, stackrel, mathrsfs, tikz}
\usepackage[colorlinks, citecolor=blue, linkcolor=blue, bookmarks=true]{hyperref}
\usepackage[nameinlink, noabbrev, capitalize]{cleveref}
\usepackage{enumitem}

\usepackage{bbm}

\usepackage{mathtools}

\makeatletter
\crefdefaultlabelformat{\leavevmode\sw@slant#2{\upshape#1}#3}
\makeatother

\newcommand{\poly}{\mathsf{poly}}

\newcommand{\eps}{\varepsilon}
\newcommand{\bit}{\{0, 1\}}

\newcommand{\final}{\mathsf{final}}

\newcommand{\distance}{\mathsf{HD}}

\newcommand{\N}{\mathsf{N}}
\newcommand{\ent}{\mathsf{H}}
\newcommand{\svec}{\mathbf{S}}
\newcommand{\snum}{\mathbf{s}}
\newcommand{\knum}{\mathbf{k}}

\newcommand{\Enc}{\mathsf{Enc}}
\newcommand{\Dec}{\mathsf{Dec}}

\newcommand{\mynull}{\mathsf{null}}

\newcommand{\ED}{\mathsf{ED}}

\newcommand{\pos}{\mathsf{pos}}

\newtheorem{lemma}{Lemma}[section]
\newtheorem{theorem}[lemma]{Theorem}
\newtheorem{claim}[lemma]{Claim}
\newtheorem{construction}[lemma]{Construction}

\newtheorem{definition}[lemma]{Definition}
\newtheorem{corollary}[lemma]{Corollary}

\newtheorem{remark}[lemma]{Remark}

\begin{document}

\allowdisplaybreaks

\begin{titlepage}
\def\thepage{}
%\title{Coding with Asymmetric Information and Improved Document Exchange Protocols}
\title{Efficient Document Exchange and Error Correcting Codes with Asymmetric Information}
\author{Kuan Cheng \thanks{ckkcdh@pku.edu.cn.\ Center on Frontiers of Computing Studies, Peking University.\  Supported in part by a start-up fund of Peking University, a Simons Investigator Award (\#409864, David Zuckerman) and NSF Award CCF-1617713.} \and Xin Li  \thanks{lixints@cs.jhu.edu.\ Department of Computer Science, Johns Hopkins University. Supported by NSF Award CCF-1617713 and NSF CAREER Award CCF-1845349.}}

\maketitle \thispagestyle{empty}

%%%%%%%%%%%%%%%%%%%%%%%%%%%%%%%%%%

\abstract 
We study two fundamental problems in communication, Document Exchange (DE) and Error Correcting Code (ECC). In the first problem, two parties hold two strings, and one party tries to learn the other party's string through communication. In the second problem, one party tries to send a message to another party through a noisy channel, by adding some redundant information to protect the message. Two important goals in both problems are to minimize the communication complexity or redundancy, and to design efficient protocols or codes.

Both problems have been studied extensively. In this paper we study whether asymmetric partial information can help in these two problems. We focus on the case of Hamming distance/errors, and the asymmetric partial information is modeled by one party having a vector of disjoint subsets $\svec=(S_1, \cdots, S_t)$ of indices and a vector of integers $\knum=(k_1, \cdots, k_t)$, such that in each $S_i$ the Hamming distance/errors is at most $k_i$. To our knowledge, no previous work has studied this problem systematically. We establish both lower bounds and upper bounds in this model, and provide efficient randomized constructions that achieve a $\min\lbrace O(t^2), O\left((\log \log n)^2\right) \rbrace $ factor within the optimum, with almost linear running time. 

We further show a connection between the above document exchange problem and the problem of document exchange under \emph{edit distance}, and use our techniques to give an efficient randomized protocol with optimal communication complexity and \emph{exponentially} small error for the latter. This improves the previous result by Haeupler \cite{haeupler2018optimal} (FOCS'19), which has polynomially large error; and that by Belazzougui and Zhang \cite{BelazzouguiZ16} (FOCS'16), which is only optimal for a limited range of parameters. Our techniques are based on a generalization of the celebrated expander codes by Sipser and Spielman \cite{sipser1996expander}, which may be of independent interests.

\iffalse
We study two interesting problems about Error Correcting Code (ECC), i.e. Coding with Asymmetric Information (CAI) for hamming distance, and  Document Exchange (DE) for edit distance. We give an efficient construction for the first problem and then use it to give an efficient construction for the second problem.

The CAI problem is in  a two-party coding scenario. Consider two parties (Alice and Bob), each holding a string s.t. the hamming distance between them are bounded. Alice wants to send a short sketch of her string to Bob to let him recover her string. 
An additional condition is that Bob knows some disjoint subsets s.t. different bounded number of errors appear and only appear in these subsets. 
But Alice doesn't know what exactly are these sets. She only knows the size of each set and the corresponding number bound for errors. 
Then is the extra but asymmetric information helpful to achieve smaller communication complexity? To our knowledge there is no literature considering this   problem before.
We give a positive answer to it through providing an efficient construction.

The DE problem is also a two-party coding problem, but instead the two strings   only have an upper bound of edit distance. 
We use our protocol for CAI and some  derandomization techniques to give a new construction for DE, achieving optimal (up to constant factor) sketch length and exponentially small failure probability. This improves the previous construction by Haeupler (FOCS'19) and that by Belazzougui \& Zhang (FOCS'16). 
\fi
\end{titlepage}

\section{Introduction}
Document exchange is a combinatorial version of the famous Slepian-Wolf problem \cite{slepian1973noiseless},  which is a fundamental problem in communication and coding theory dating back to 1973.
It was then studied by Orlitsky \cite{orlitsky1991interactive} and subsequently named and also studied by Cormode et. al.\ \cite{CormodePSV00}.  
Here, two parties Alice and Bob each holds a string (document) $x$ and $y$, and the goal is for one party to learn the other party's string with the least amount of communication possible. For simplicity, let us assume that both $x$ and $y$ have $n$ bits. If $x$ and $y$ can be arbitrary strings, then it is clear that in the worst case the communication needs at least $n$ bits, i.e., sending one party's string to the other party. However, in practice this is often not the case, and $x$ and $y$ can actually be close in some sense. For example, Alice and Bob may be two uses holding different versions of some original document, where $x$ and $y$ are obtained after some edits of a string $z$. If the number of edits is limited, then it is possible for one party to learn the other party's string with significantly less amount of communication. In this paper, we focus on the case where the strings have binary alphabet.

More generally and formally, the document exchange problem can be described as follows. Alice and Bob each has an $n$-bit string $x$ and $y$, and the distance between $x$ and $y$, $D(x, y)$ is upper bounded by some number $k$. Here the distance $D$ can be any measure of interests. Now, the first goal here is to minimize the communication complexity as a function of $n$ and $k$. In addition, it is also an important goal to keep the protocol \emph{efficient}, i.e., we would like the communication protocol to run in polynomial time of $n$.

There has been a lot of work on the document exchange problem \cite{orlitsky1991interactive, barbara1988exploiting, BarbaraL91, AbdelA94, CormodePSV00, orlitsky2001practical, suel2004improved, irmak2005improved, Jowhari2012EfficientCP, BelazzouguiZ16, cheng2018deterministic, haeupler2018optimal, cheng2019block}. While Orlitsky \cite{orlitsky1991interactive} established some upper and lower bounds on the communication complexity of general ``balanced" measures $D(x, y)$, as well as exponential time protocols that can achieve the optimal communication, efficient protocols in subsequent works have been mostly focusing on the two natural cases where $D(x, y)$ is either the Hamming distance or the edit distance. In the former, the distance is measured by how many bits in $x$ and $y$ are different at the corresponding locations, while in the latter the distance $\ED(x, y)$ is measured by the minimum number of insertions, deletions, and substitutions to transform one string into another.\ Both distances are metrics, and edit distance strictly generalizes Hamming distance. 

For both Hamming distance and edit distance, it is known that if $D(x, y) \leq k$, then the optimal communication complexity in the document exchange problem is $\Theta(k \log (n/k))$, and this can be achieved by a deterministic one-round protocol running in exponential time. The situation of efficient protocols however is different for these two measures. For Hamming distance, we have many efficient, deterministic one-round protocols with optimal communication complexity $\Theta(k \log (n/k))$, using checksum decoding of error-correcting codes such as Algebraic Geometry codes \cite{hoholdt1998algebraic}, BCH codes \cite{bose1960further, hocquenghem1959codes}, etc. For edit distance, except for the exponential time deterministic one-round protocol in \cite{orlitsky1991interactive} which achieves optimal communication complexity, for a long time only efficient randomized one round protocols with sub-optimal communication complexity are known. These include the work of Irmak et al. \cite{irmak2005improved} with communication complexity $O(k \log(\frac{n}{k}) \log n)$, the work of Jowhari \cite{Jowhari2012EfficientCP} with communication complexity $O(k \log^2 n \log^* n)$, the work of Chakraborty et al. \cite{Chakraborty2015LowDE} with communication complexity $O(k^2 \log n)$, and the work of Belazzougui and Zhang \cite{BelazzouguiZ16} with communication complexity $O(k (\log^2 k+\log n))$. In particular, the protocol in \cite{BelazzouguiZ16} has asymptotically optimal communication complexity for $k =2^{O(\sqrt{\log n})}$, with success probability $1-1/\poly(k \log n)$. 

In 2018, Cheng et. al. \cite{cheng2018deterministic}, and Haeupler \cite{haeupler2018optimal} independently gave an efficient, deterministic one-round protocol with communication complexity $O(k \log^2 (n/k))$. Finally, Haeupler \cite{haeupler2018optimal} gave the first efficient randomized one-round protocol with optimal communication complexity $O(k \log (n/k))$. However, his protocol only succeeds with probability $1-1/\poly(n)$.

Document exchange is closely related to the (even more) fundamental problem of error correcting codes.\ The goal of an error correcting code is to ensure that one party can successfully send information to another party, despite errors caused by the communication channel. In this setting, the first party (Alice) runs an encoding algorithm that turns a message of $m$ bits into a codeword of $n$ bits, and sends the codeword to the second party (Bob) through a channel.  Bob then tries to recover the message by running a decoding algorithm. Similar to document exchange, there are also two important goals here. First, one wants to keep $n-m$ (the redundancy of the codeword) to be as small as possible, or alternatively, to keep $m$ (the message length) to be as large as possible. Second, one needs both the encoding and decoding to be efficient, i.e., run in polynomial time of $m$.

There has been extensive study on error correcting codes, which we will not be able to completely survey here. Again, the channel error can have several different models, and the most studied are Hamming errors and edit errors. For both cases, assuming $k$ is an upper bound on the number of errors, then it is known that the optimal message length one can achieve (with possibly exponential time encoding/decoding) is $m=n-\Theta(k \log(n/k))$. For Hamming errors, again we have efficient constructions matching this bound, based on Algebraic Geometry codes \cite{hoholdt1998algebraic}. For edit errors the constructions are far behind, and for a long time we only have asymptotically optimal constructions for the two extreme cases of $k=1$ \cite{Levenshtein66} and $k=\alpha n$ for some small constant $\alpha>0$ \cite{796406}. A recent line of works \cite{7835185, 7541373, BukhV16, haeupler2017synchronization, HS17c, cheng2018deterministic, haeupler2018optimal} achieved significant progress on this problem. In particular, Cheng et.\ al.\ \cite{cheng2018deterministic}, and Haeupler \cite{haeupler2018optimal} independently gave an efficient code with $m=n-O(k \log^2(n/k))$. Cheng et.\ al.\ \cite{cheng2018deterministic} further gave an efficient code with $m=n-O(k \log n)$, which is optimal for $k \leq n^{1-\alpha}$ where $\alpha>0$ is any constant.

The connection between document exchange and error correcting codes is demonstrated by the notion of systematic error correcting codes. These are codes where a codeword is simply the message followed by some redundant information called the \emph{checksum}. Given such a code, the checksum can be used as the information sent in a document exchange protocol. Conversely, given a one round document exchange protocol, one can use a standard error correcting code on the information sent and use it as the checksum in a systematic error correcting code.

In all previous works, Alice and Bob have symmetric information---they both know that their string is within distance $D(x, y) \leq k$ to the other party's string, or the total number of errors in the received codeword is at most $k$. However, in many practical situations, each party may have some additional partial information that is not known to the other party. For example, in document exchange, if Bob has made edits in some specific parts of the original document, then even without carefully tracking the edits, Bob has some partial information of where the differences can happen. This information is not necessarily known to Alice. In another situation, suppose Alice sends a long string to Bob by Internet routing, then this string may be broken into several parts and transmitted to Bob through different channels. These channels may have different behavior and introduce different numbers of errors. While it is reasonable that both parties know the parameters of all channels, due to the routing process Alice may not know which channels her parts are sent through. On the other hand, Bob can learn these information by observing the received parts. Thus Bob will have some partial information about the numbers of errors in specific parts of the received string, which is not known to Alice. The fist example applies to document exchange and the second example applies to error correcting codes. One can now ask the following natural question, which is the focus of this paper. \\ 

\noindent \textbf{Question:} \emph{Can we use these asymmetric information to reduce the communication complexity in document exchange or the redundancy in error correcting codes, while still designing efficient protocols or codes?}  \\

Towards answering this question, we first formally define our model. %We first introduce the background  about the coding  with asymmetric information (CAI) problem. Then we introduce the background  of Document Exchange (DE) problem. Then we provide our results and technique overview, during which we describe  how to use CAI to construct a better DE protocol. 
%The paper first studies an interesting coding (for hamming distance) with asymmetric information problem, and then uses it to construct a Document Exchange protocol (for edit distance)  with sketch length optimal up to a constant factor and with exponentially small failure probability. The first problem is also of independent interests as we will explain in the following sections.

\subsection{The Model of Asymmetric Information}
In this paper we focus on Hamming distance/Hamming errors in the model of asymmetric information. To model the asymmetric information, we assume that one party has some additional information of where the differences/errors can happen. More formally, we use a vector of disjoint subsets $\svec=(S_1, \cdots, S_t)$ to indicate the positions where the differences/errors can happen, and a vector of integers $\knum=(k_1, \cdots, k_t)$ to indicate the upper bounds on the numbers of differences/errors in each set $S_i$. For each $S_i$, let $s_i$ denote the size of $S_i$, i.e., $s_i=|S_i|$. We also use $\snum$ to indicate the vector $\snum=(s_1, \cdots, s_t)$. We assume the parameters $(\snum, \knum, t)$ are known to both parties, and that (without loss of generality) $k_1 \geq k_2 \geq \cdots \geq k_t$. 

\begin{definition}($(\snum, \knum, t)$ Asymmetric Document Exchange)
There are two parties Alice and Bob. Alice has a string $x\in \{0,1\}^n$ and Bob has a string $y\in \{0,1\}^n$. Both parties know  $(\snum, \knum, t)$. In addition, Bob knows a vector of disjoint subsets $\svec=(S_1, \cdots, S_t)$ where $\forall i, S_i \subseteq [n]$ and $|S_i|=s_i$. That is, within each set $S_i$, the Hamming distance between $x$ and $y$ is at most $k_i$. One party tries to learn the string of the other party.
\end{definition}

\begin{definition}($(\snum, \knum, t)$ Asymmetric Error Correcting Code)
There are two parties Alice and Bob.\ Both parties know  $(\snum, \knum, t)$.\ Alice encodes a message of $m$ bits into a codeword of $n$ bits, using a function $\Enc: \bit^m \to \bit^n$ and sends it to Bob. Bob knows a vector of disjoint subsets $\svec=(S_1, \cdots, S_t)$ where $\forall i, S_i \subseteq [n]$ and $|S_i|=s_i$.\ That is, within each set $S_i$, there are at most $k_i$ Hamming errors in the received codeword. Bob uses a function $\Dec: \bit^n \to \bit^m$ to recover the message. %Alice knows the size of the subets $\snum=(s_1, \cdots, s_t)$ and the numbers of errors $\knum=(k_1, \cdots, k_t)$.
\end{definition}

We require the protocol or code to succeed for \emph{every} possible vector of disjoint subsets $\svec=(S_1, \cdots, S_t)$ with $|S_i|=s_i, \forall i$, and for \emph{every} possible distance/error pattern that is consistent with $\svec=(S_1, \cdots, S_t)$ and $\knum=(k_1, \cdots, k_t)$. %We call them an $(\snum, \knum, t)$ asymmetric document exchange (DE) protocol or an $(\snum, \knum, t)$ asymmetric error correcting code (ECC). 

We consider both deterministic and randomized protocols/codes. In the case of randomized solutions, we assume that the two parties have shared randomness, as is standard in all previous works. In the case of error correcting codes, we further assume that the channel errors do not depend on the shared randomness.

Our model is quite general in capturing asymmetric information. A naive solution is to simply ignore the extra information, and apply a document exchange protocol or error correcting code for $k=\sum_{i=1}^{t} k_i$ Hamming distance or Hamming errors. However, our goal here is to see if the extra information can be used to design better protocols or codes. Another natural strategy for the document exchange problem, is for Bob to first send the descriptions of $\svec$ to Alice, and they can then run a protocol on each set $S_i$. However, this strategy can result in a significant amount of communication, e.g., $\sum_{i=1}^{t} s_i \log n$, which can be even larger than $n$. In some special situations, a set $S_i$ may be a continuous block in the string, and it suffices to just send the starting and ending index, using $2 \log n$ bits. If all sets $S_i$ are of this form, then the total number of bits required is $2 t \log n$. Even this number can be large when the number of sets $t$ is large. We also stress that in our model and all results, each set $S_i$ does not need to be a continuous block. A final simple strategy is to try to form a large continuous block which includes several $S_i$'s, but this can increase the size of the sets significantly and thus also results in a penalty on the communication complexity. 

\begin{remark}
In the asymmetric document exchange, it may seem unreasonable to assume that Alice knows the vectors $\snum, \knum$. However, this is without loss of generality up to a small loss in communication complexity and communication rounds. Basically, Bob can first send these two vectors to Alice. This only takes one round and the number of bits sent by Bob is $O(\sum_{i=1}^{t}(\log k_i+\log s_i))$, while the number of bits needed to distinguish all possible error patters is at least $\sum_{i=1}^{t} \log \binom{s_i}{k_i}$. The former is always within a constant factor to (and in most cases smaller than) the latter. 
\end{remark}

\paragraph{Related previous works.} While document exchange and error correcting codes with asymmetric information are natural questions, to our knowledge they have not been studied systematically. The only previous work we found is the work of Belazzougui and Zhang \cite{BelazzouguiZ16}, which studies a special case of our model with $t=1$, i.e., Bob's extra information only has one subset $S$ with $|S|=s$. They use entirely different techniques to give a document exchange protocol with sub-optimal communication complexity $O(k(\log s+\log(1/\eps)))$, where Bob can learn Alice's string with success probability $1-\eps$. 

However, there are a large body of works on a related topic \cite{adler2001protocols, LaberH02, GhodsiS01, WatkinsonAF01, AdlerDHP06, AGHRW18}, which study the problem of source coding/data compression with asymmetric information. In this setting, the decoder has some prior distribution $\mu$ not known to the encoder, and the encoder tries to send a set of items drawn independently from the distribution to the decoder, using the smallest number of bits as possible. The problem we study here, on the other hand, focuses on error correction. While there are similarities between these two problems, they are also fundamentally different. For example, all the efficient algorithms in these prior works run in time polynomial in the size of the support $\mu$. This is prohibitive for our purpose since this number is already exponentially large. 

We note that source coding and error correction are the two most important applications of information theory. Thus given the abundant works on source coding/data compression with asymmetric information, we believe a systematic study of document exchange and error correcting codes with asymmetric information is also an important direction.

\subsection{Our Results}
We provide both lower bounds and upper bounds for document exchange and error correcting codes with asymmetric information. To simplify the presentation, we first define some quantities. Given two vectors $\snum=(s_1, \cdots, s_t)$ and  $\knum=(k_1, \cdots, k_t)$, we define $\ent(\snum, \knum)=\log \left (\prod_{i=1}^{t} \left (\sum_{j=0}^{k_i}\binom{s_i}{j} \right ) \right )=\sum_{i=1}^{t} \log \left (\sum_{j=0}^{k_i}\binom{s_i}{j} \right )$.\ Similarly, for two integers $s$ and $k$ with $s \geq k$, we define $\ent(s, k)=\log \left (\sum_{j=0}^{k}\binom{s}{j} \right )$.

Note that if $\forall i, s_i \geq 2k_i$ and $s \geq 2k$, then $\ent(\snum, \knum) =\Theta( \sum_{i=1}^{t} k_i \log(s_i/k_i))$ and $\ent(s, k) =\Theta( k \log (s/k))$. Recall that $k=\sum_{i=1}^{t} k_i$ and $s=\sum_{i=1}^{t} s_i \leq n$, hence $\ent(\snum, \knum) \leq \ent(n, k)$. We have the following theorem.

\begin{theorem}\label{thm:lowerDE}
In an $(\snum, \knum, t)$ asymmetric DE problem, we have
\begin{itemize} 
\item Suppose  Alice learns Bob's string, then any deterministic protocol has communication complexity at least $\ent(n, k)$, and any randomized protocol with success probability $\geq 1/2$ has communication complexity at least $\ent(n, k)-1$. %This holds even if Alice knows $\snum$ and $\knum$.
\item Suppose  Bob learns Alice's string, then any randomized protocol with success probability $\geq 1/2$ has communication complexity at least $\ent(\snum, \knum)-1$. Furthermore if $\forall i, s_i \geq 2 k_i$, then any one round deterministic protocol has communication complexity at least $\ent(n, k)$. %This holds even if Alice knows $\snum$ and $\knum$.
\end{itemize}
\end{theorem}

This theorem tells us the following important things: First, Bob's extra information is only useful for him to learn Alice's string, but not useful in the other direction. Second, in the case of a one round protocol for Bob to learn Alice's string, for a wide range of parameters (i.e., when $\forall i, s_i \geq 2 k_i$), Bob's extra information is only useful in randomized protocols. 

For upper bounds, we note that there are efficient deterministic protocols to meet the bound $\ent(n, k)$, based on algebraic geometry codes. To meet the bound $\ent(\snum, \knum)$, there is also a simple one round randomized protocol: Alice hashes her string $x$ using a random hash function, and Bob enumerates all possible strings to find the one with the correct hash value. It's easy to see that this protocol succeeds if there is no hash collision, which happens with high probability if the hash function outputs some $O(\ent(\snum, \knum))$ bits. However, this protocol runs in exponential time, and our main result is an \emph{efficient} protocol that gets close to this bound.

To state our main theorem, we define another quantity $\chi(s, k, t) \in \N$: first partition the interval $[2, n]$ into disjoint subintervals $\{I_j=[2^{10^{j-1}}, 2^{10^{j}})\}$, starting from $j=1$. Then, for every $i \in [t]$, put $s_i/k_i$ into the corresponding subinterval. $\chi(s, k, t)$ is defined to be the number of subintervals $I_j$ which contain at least one $s_i/k_i$. We now have the following theorem.

\begin{theorem}\label{thm:upperDE}
In an $(\snum, \knum, t)$ asymmetric DE problem, suppose that $\forall i, s_i \geq 2 k_i$. There is an efficient randomized one round protocol for Bob to learn Alice's string, with communication complexity $O(\chi(s, k, t)^2 \ent(\snum, \knum))$ and error probability $2^{-\Omega(k_t)}+\frac{1}{\poly(s)}$. The protocol runs in time $\widetilde{O}(n)$.
\end{theorem}   

Note that $\chi(s, k, t) \leq t$ and $\chi(s, k, t) \leq \log \log n$, so the above theorem immediately gives the following two corollaries.

\begin{corollary}\label{cor:upperDE1}
In an $(\snum, \knum, t)$ asymmetric DE problem, suppose that $\forall i, s_i \geq 2 k_i$. There is an efficient randomized one round protocol for Bob to learn Alice's string, with communication complexity $O(t^2 \ent(\snum, \knum))$ and error probability $2^{-\Omega(k_t)}+\frac{1}{\poly(s)}$. The protocol runs in time $\widetilde{O}(n)$.
\end{corollary}   

\begin{corollary}\label{cor:upperDE2}
In an $(\snum, \knum, t)$ asymmetric DE problem, suppose that $\forall i, s_i \geq 2 k_i$. There is an efficient randomized one round protocol for Bob to learn Alice's string, with communication complexity $O((\log \log n)^2 \ent(\snum, \knum))$ and error probability $2^{-\Omega(k_t)}+\frac{1}{\poly(s)}$. The protocol runs in time $\widetilde{O}(n)$.
\end{corollary}   

In particular, Corollary~\ref{cor:upperDE1} implies that if $t$ is a constant, then we have a one round protocol with asymptotically optimal communication complexity, while Corollary~\ref{cor:upperDE2} gives a one round protocol with communication complexity optimal up to an additional $(\log \log n)^2$ factor. Both protocols run in near linear time. We also note that the simple strategy of ignoring the extra information can result in communication complexity $\Omega(\ent(\snum, \knum) \log n)$ in the worst case.

Similarly, we have both lower bounds and upper bounds for error correcting codes with asymmetric information. The first theorem shows that such information is only useful for a randomized code.

\begin{theorem}\label{thm:lowerECC}
In an $(\snum, \knum, t)$ asymmetric ECC problem, if $\forall i, s_i=|S_i| \geq 2 k_i$, then any deterministic code must have distance at least $2k+1$. In particular, this means $m \leq n-\ent(n, k)$. Furthermore, any randomized code with success probability $\geq 1/2$ must have message length $m \leq n-\ent(\snum, \knum)+1$. 
\end{theorem}

Again, a code with randomized encoding and exponential time deterministic decoding can achieve message length $m = n-O(\ent(\snum, \knum)$. We design an efficient code that comes close to this.

%There is a constant $c>2$ such that the following holds.
\begin{theorem}\label{thm:upperECC}
In an $(\snum, \knum, t)$ asymmetric ECC problem, suppose $\forall i, s_i \geq 2 k_i$.\ There is an efficient code with randomized encoding and deterministic decoding, which has message length $m=n-O(\chi(s, k, t)^2 \ent(\snum, \knum))$ and error probability $2^{-\Omega(k_t)}+\frac{1}{\poly(s)}$. In particular, the message length can be $\max\{n-O(t^2 \ent(\snum, \knum)), n-O((\log \log n)^2 \ent(\snum, \knum))\}$, and the running time is $\widetilde{O}(n)$.
\end{theorem} 

%In all our protocols and codes, the condition $s_i \geq 2 k_i$ can be replaced by $s_i \geq (1+\alpha) k_i$ for any constant $\alpha>0$.

Next we show that we can design efficient document exchange protocols with asymptotically optimal communication complexity in a special case, roughly when $\snum, \knum$ are geometric progressions.

\begin{theorem}
\label{thm:DEHspecialintro}
%For any constants $\alpha, \beta \in (0,1)$, there is an efficient one-round protocol for the $(\mathbf{s}, \mathbf{k}, t)$ asymmetric DE problem, where $ s_i =  k2^{\Theta(i)}, k_i = \max\{k/2^{\Theta(i)}, k^{1-\alpha-\beta}\} \leq s_i/40 $, $t  = O(k^{\beta/2} \sqrt{\log k})$.  The communication complexity is $O(k)$ and the error probability is $2^{-\Theta(k^{1-\alpha} ) }$. The protocol runs in time $\widetilde{O}(n)$.

There is an efficient randomized one-round protocol for every $(\mathbf{s}, \mathbf{k}, t)$ asymmetric DE problem, where $ s_i =  k2^{\Theta(i)} ,  k_i = \max\{k/2^{\Theta(i)}, \Theta(\frac{k}{ t^2 \log  \frac{n}{k}  }) \} \leq s_i/40 $.  The communication complexity is $O(k)$ and the error probability is $2^{-\Omega(k / \log \frac{n}{k} ) } $.

\end{theorem}

We show that the problem of document exchange under edit distance can be reduced to the special case above, and thus we obtain the following theorem. %with optimal communication complexity and exponentially small error.
\begin{theorem}

%There is an efficient one-way document exchange protocol, for every $k = \Omega(\log \log  \frac{n}{k})^{1.2}$, with sketch length $O(k \log \frac{n}{k})$,   success  probability $1 -   2^{-\Omega(k^{0.9})}$.
%
%The encoding   time is $\tilde{O}(n^2)$. The decoding  time is $\tilde{O}( \min(n^3, k^2 n^2) )$.

%There is an efficient randomized one-round protocol for the DE problem with edit distance at most $k$. The communication complexity is $O(k \log \frac{n}{k})$ and the error probability is $2^{-k^{\Omega(1)}}$. In addition, for any constant $\beta \in (0,1)$ and $k = \Omega(  \frac{\log n}{\sqrt{\log \log n}})^{2/\beta}$, the error probability is $2^{-\Omega(k^{1-2\beta})}$.

There is an efficient randomized one-round protocol for the DE problem with edit distance at most $k$. The communication complexity is $O(k \log \frac{n}{k})$ and the error probability is $\min\{2^{-\Theta(k/  \log^3 \frac{n}{k} ) }, 1/\poly(n) \} $.

\end{theorem}
 
We also have both lower bounds and upper bounds for document exchange where both parties have some asymmetric partial information, represented as a vector of disjoint subsets. For the clarity of presentation we omit the results here, and refer the reader to Section~\ref{sec:twoside} for details. 
%We consider the case where the total size of all subsets from both parties is at most $n$. 

\subsection{Technique Overview}  
Our lower bounds follow from relatively simple information theoretic arguments, so here we only provide an informal outline of our protocols. We start with the asymmetric document exchange for Hamming distance. Recall that the asymmetric information is in the form of $\svec=(S_1, \cdots, S_t)$ and $\knum=(k_1, \cdots, k_t)$, where $\forall i, |S_i|=s_i$ and the Hamming distance within $S_i$ is at most $k_i$. We assume $\forall i, s_i \geq 2k_i$, and without loss of generality that $k_1 \geq k_2 \cdots \geq k_t$.
%\subsubsection{DE and ECCwith asymmetric information}
 
\paragraph{The protocol for one set.} Our starting point is the simplest case where $t=1$, i.e. there is only one set $S$ of size $s$ and the Hamming distance in $S$ is at most $k$. In this case our goal is to give an efficient one round protocol with communication complexity $O(k \log \frac{s}{k})$. If $s=n$ then this can be achieved by using a systematic algebraic geometry code or an expander code \cite{sipser1994expander}. We will use the latter and we briefly review the application of expander codes in document exchange.

To run the protocol, the two parties choose a bipartite expander graph $G: [n] \times [d] \to [m]$. Alice associates her string $x$ with the $n$ vertices on the left,  and computes a string $z$ of length $m$ as follows: For every $i\in [m]$, let  $z_i = \bigoplus_{j\in \Gamma^{-1}(i)} x_j$, where $\Gamma^{-1}(i)$ is the set of neighbors of the right vertex $i$ in the expander. The string $z$ consists of a sequence of parity checks of $x$, and is then sent to Bob.%Each such equation is called a check.

To recover $x$, Bob starts out with $\tilde{x}=y$ as his current version of $x$, and maintains another string $z' \in \bit^m$ using the same approach as above, except replacing the string $x$ by $\tilde{x}$, i.e., $z'$ consists of a sequence of parity checks of $\tilde{x}$. $z$ and $z'$ will differ in several coordinates, and Bob will gradually modify $\tilde{x}$ into $x$ by flipping some bits in $\tilde{x}$ according to the parity checks. This process is known as belief propagation, and works as follows. Bob keeps finding a bit in $\tilde{x}$ such that by flipping this bit, the Hamming distance between $z'$ and $z$ decreases by at least one. Bob flips this bit and updates $\tilde{x}$ and $z'$ correspondingly. Bob stops when $z'=z$, at which point $x=\tilde{x}$ and he has successfully recovered $x$.  %The decoder Bob has his corrupted version of the message $\tilde{x}$ and also the correct sketch $z$. He can also compute a sketch $z'$ of $\tilde{x}$, which has some difference from $z$. We call the corresponding check of a different bit as an unsatisfied check. Otherwise it's called satisfied.   

%The belief propagation decoding is as follows. The decoder keeps finding a bit in the current message s.t. this bit has more satisfied checks than unsatisfied ones, and then flipping it. He stops once this kind of bit cannot be found.

For the analysis, we use the set $R \subseteq [n]$ to denote the coordinates where $x$ and $\tilde{x}$ are different. We say the $i$'th parity check bit is satisfied if $z_i=z'_i$, and unsatisfied otherwise. Let the number of satisfied and unsatisfied checks in $\Gamma(R)$ (the neighbors of $R$) be $\mathsf{s}$ and  $\mathsf{u}$.
Assume the graph has good expansion, i.e. $|\Gamma(R)| = \mathsf{s} + \mathsf{u} \geq 0.9 d|R|$, and note that in $\Gamma(R)$, each satisfied check has at least two neighbors in $R$. Thus $2\mathsf{s} + \mathsf{u} \leq d|R|$. By the two inequalities, we deduce $\mathsf{u} \geq 0.8 d |R|$ and thus at least one left vertex has more unsatisfied parity checks as neighbors than satisfied parity checks, and Bob can flip this bit. The analysis holds as long as the expansion of the set $R$ is guaranteed. Note that the number of unsatisfied checks is strictly decreasing in the process, thus $|R|$ can never be more than $1.25 k$, since otherwise this will induce more than $dk$ unsatisfied checks, but at the beginning there are at most $dk$ unsatisfied checks. Therefore, we only need to guarantee the expansion of all $R \subseteq [n]$ with $|R| \leq 1.25k$, and a random graph with $m=O(k \log \frac{n}{k})$ and $d=O(\log \frac{n}{k})$ satisfies this property with high probability.
 
Going back to the case where $s< n$, the first issue is that we can't afford to use an expander which has good expansion for all subsets $R$ as before, since this will make $m=\Omega(k \log \frac{n}{k})$. To fix this, we instead just require the expansion to hold for all subsets $R \subseteq S$ with $|R| \leq 1.25k$. Now, a random graph with $m=O(k \log \frac{s}{k})$ and $d=O(\log \frac{s}{k})$ satisfies this property with high probability, and both parties can generate the same expander by using the shared randomness. Similarly, when recovering $x$ Bob will always look for a bit in $S$ to flip. The analysis is now similar to the standard case and this gives the protocol for the case of $t=1$. 
%However, this raises another issue in the process of recovering $x$, namely that the standard belief propagation approach may flip some bits outside of $S$, while the expansion only holds for subsets inside $S$. To address this, we modify the process such that Bob always looks for a bit in $S$ to flip. The analysis is now similar to the standard case and this gives the protocol for the case of $t=1$. 

\paragraph{The protocol for two sets.}
We now consider the case with $t>1$. Our goal is to design an efficient one round protocol with communication complexity close to $\ent(\mathbf{s}, \mathbf{k})$.

The first idea may be to take the union of all $S_i, i\in [t]$ as one set $S$, and the Hamming distance in $S$ is at most $k = \sum_{i\in [t]} k_i$. Now we can use the protocol for $t=1$ described before. However, in this case the communication complexity will be $O(\ent(s, k))$, which may not be close to $\ent(\snum, \knum)$. For example, consider the case where $t = 2$, $k_1 = n^{0.1}, s_1  = 10n^{0.1}$, $k_2 = 10, s_2 = 0.1 n$. A direct computation indicates $\ent(s, k) = \Omega(\ent(\snum, \knum) \log n)=\omega(\ent(\mathbf{s}, \mathbf{k}))$. It also appears hard to improve this if we just use a single expander graph, since the decoding requires good expansion for all possible subsets of errors during the belief propagation, which can potentially be all possible subsets of size $\Omega(k)$. This forces the right hand size of the graph to be $\ent(s, k)$.

To overcome this difficulty, our idea is to use more than one expander codes. Towards this, our main observation is that, the issue with the above example is due to the following fact: for some $i\in [t]$, $k_i$ is large while $s_i $ is small, but for some other $i$, $k_i$ is small while $s_i$ is large. Indeed, in the case of $t=2$, there are two good situations where $\ent(s, k)=O(\ent(\snum, \knum))$:

\begin{enumerate}
\item $k_1$ and $k_2$ are roughly the same, i.e., $k_1=\Theta(k_2)$. In this case we have $\ent(s, k)=\Theta((k_1+k_2) \log \frac{s_1+s_2}{k_1+k_2})=\Theta(k_1 \log \frac{s_1}{k_1}+k_2 \log \frac{s_2}{k_2})=\Theta(\ent(\snum, \knum))$.

\item $\log \frac{s_1}{k_1}$ and $\log \frac{s_2}{k_2}$ are roughly the same, i.e., $\log \frac{s_1}{k_1}=\Theta(\log \frac{s_2}{k_2})$. In this case we also have $\ent(s, k)=\Theta(\ent(\snum, \knum))$.
\end{enumerate} 

Our protocol will exploit both of these good cases. We first illustrate this with a protocol for the case of $t=2$. Our idea is to reduce the number $k_1$ (recall that $k_1 \geq k_2$) to be roughly the same as $k_2$ (which is unnecessary if $k_1$ and $k_2$ are already roughly the same at the beginning). In other words, we will first reduce the Hamming distance in $S_1$ from $k_1$ to at most $c k_2$, if $k_1 > c k_2$ for some constant $c>1$. It is not immediately clear why this is feasible, since Alice does not know the subset $S_1$. Additionally, we need to make sure the communication complexity of this step is not too large.

We achieve this by using an expander code based on a bipartite expander $G: [n] \times [d] \rightarrow [m]$ such that for all sets $R \subseteq S_1$ with $|R| \in [ c k_2, 1.6 k_1]$, the set $R$ has good expansion, i.e., $|\Gamma(R)| \geq 0.9d |R|$. The expander is again generated by shared randomness, and we show that we can choose $d = O(\log \frac{s_1}{k_1}), m = O(k_1\log \frac{s_1}{k_1})$ and the graph satisfies the property with high probability.  Alice will again compute the parity checks $z$ and send it to Bob.

Now Bob will apply the same method as before: start with $\tilde{x}=y$ and keep finding a bit in $S_1$ with more unsatisfied parity checks as neighbors than satisfied parity checks. Bob flips this bit and continues doing this until no such bit can be found. Since the number of unsatisfied parity checks keeps decreasing, the process will end in a finite number of steps. We claim that when it ends, the Hamming distance in $S_1$ is at most $c k_2$.\ This effectively reduces the Hamming distance in $S_1$.  

The main issue in the analysis here is that the different bits between $x$ and $y$ are not entirely in $S_1$, and this may cause problems in belief propagation. However, our observation is that when $k_1$ is much larger than $k_2$, the effect of $k_2$ can mostly be ignored. More specifically, let $R$ be the set of left vertices which correspond to the different bits between $x$ and $y$ in $S_1$, and $R_2$ be the set of left vertices which correspond to the different bits  in $S_2$. Thus $|R_2| \leq k_2$. Let the number of satisfied and unsatisfied checks in $\Gamma(R_1)$ be $\mathsf{s}$ and  $\mathsf{u}$. As long as $|R| \in [ c k_2, 1.6 k_1]$, we have $|\Gamma(R)| = \mathsf{s} + \mathsf{u} \geq 0.9 d|R|$, and $ 2\mathsf{s} + \mathsf{u} \leq d|R| + d|R_2|  \leq (1+\frac{1}{c})|R|$. Combining these inequalities, we can still deduce $u \geq 0.7 d |R|$,  by setting $c = 10$.  Hence there must exist a bit in $S_1$ to flip. Since the number of unsatisfied checks decreases strictly, the size $|R|$ in the process can never be larger than $1.6 k_1$. This is because otherwise there will be at least $1.12 dk_1$ unsatisfied checks, while at the beginning there are only at most $(1+1/c)dk_1  = 1.1 dk_1$ unsatisfied checks. Thus when this process stops, we must have $|R| \leq c k_2$. At this point, we can use the protocol for one set together with another expander graph to finish the job, by considering the set $S=S_1 \cup S_2$ which has Hamming distance at most $(c+1)k_2$. The total communication complexity is $O(k_1\log \frac{s_1}{k_1})+O(k_2 \log \frac{s_1+s_2}{k_2})=O(\mathsf{H}( \mathbf{s}, \mathbf{k}))$. 

%Note that after reducing the number of errors in $S_1$ to be $O(k_2)$, we can call the protocol for one-set to finish the job, having total communication complexity $\mathsf{H}( \mathbf{s}, \mathbf{k})$ for the two-set case.

\paragraph{The protocol for arbitrary $t$.} We now generalize the above protocol to arbitrary $t$. Recall that $k_1 \geq k_2 \geq \cdots \geq k_t$. Our idea is to use the above protocol of reducing Hamming distance repeatedly, while going through the index from $1$ to $t$. More formally, we use $i'$ to denote the current index and $k'$ to denote an upper bound of the Hamming distance in $\cup_{j\in [i']} S_j$ after possible steps of reducing distance. We start with $i'=0, k' = 0$ and repeat the following: find the first index $i > i'$ s.t.  the current Hamming distance in $\cup_{j\in [i]} S_j$  is much larger than the Hamming distance in $\cup_{j= i+1}^t S_j$, i.e., 
\begin{align}
\label{ieq:findingi}
k'+ \sum_{j=i'+1}^{i} k_j > c \sum_{j=i+1}^t k_j = k''.
\end{align}
Then we reduce the Hamming distance in $\cup_{j \in [i]} S_i$ to at most $k''$ by using the two set protocol described before, regarding $ \cup_{j\in [i]} S_j$ as one set and $ \cup_{j=i+1}^{t} S_j$ as the other set.\ We now update $k' = k'', i' = i$ and continue the process. Finally, the Hamming distance in $S=\cup_{j\in [t]} S_i$ will be reduced to at most $(c+1)k_t$, and we apply the one set protocol for $S$ to finish the job.

The correctness follows from the correctness of the one set protocol and the two set protocol. The main thing left is to bound the communication complexity. 
Note that except the first iteration, in each subsequent iteration $i'$ will be updated to $i'+1$. Thus the number of bits Alice sends in this step is $m = O\left( (k' + k_i) \log \frac{\sum_{j\in [i]}s_j }{ k'+k_i } \right)$. We show that this is always $O(t\mathsf{H}(\mathbf{s}, \mathbf{k}))$ by using the bound on $k'$, the fact that $k_1 \geq k_2 \geq \cdots \geq k_t$, and $k_i \leq s_i/2, \forall i\in [t]$.
Thus the total communication complexity is $O(t^2 \mathsf{H}(\mathbf{s}, \mathbf{k}))$. Note that this is a one round protocol since only Alice sends out information. 

Finally, we can get further improvement by grouping some sets together. Specifically, we divide the interval $[2, n]$ into disjoint subintervals $I_j = [2^{10^{j-1}}, 2^{10^{j}} ), j = 1, \ldots, O(\log \log n)$ and put each subset $S_i$ into one interval according to the number $s_i/k_i$. Whenever two subsets $S_i$ and $S_j$ are in the same interval, we have $\log(s_i/k_i)=\Theta(\log(s_j/k_j))$ and thus we can consider $S_i \cup S_j$ as one set with Hamming distance $k_i+k_j$, without changing the communication complexity much. Now, taking the union of all subsets in the same interval to be one subset reduces the number of subsets to $\chi(s, k, t)$, and applying our protocol results in communication complexity $O(\chi(s, k, t)^2 \mathsf{H}(\mathbf{s}, \mathbf{k}))$.

\paragraph{ECC with asymmetric information.} The protocol for document exchange can be used to construct an error correcting code. We do this by first estimating the length of the redundant information. Let $m_0$ be the communication complexity of the $(\mathbf{s}, \mathbf{k}, t)$ DE protocol for message length $n$. We choose an asymptotically good code $C_0$ with message length $m_0$ and codeword length $n_0$, which corrects $k$ errors. The actual message length of our code will be $n- n_0$. On  input message $x$, we run Alice's DE protocol on $x\circ \mathbf{0}$ where $\mathbf{0} =0^{n_0}$ to get $z \in \bit^{m_0}$. Then we encode $z$ by $C_0$ and the final codeword is $x\circ C_0(z)$. 
To decode, one first recovers $z$ by running the decoding algorithm of $C_0$ on the part $C_0(z)$. Then we run Bob's DE protocol using $z$, and by replacing the $C_0(z)$ part with $0^{n_0}$. The correctness follows from the code $C_0$ and the DE protocol. %is easy to see since $C_0$ can resist $k$ errors so we can get the sketch. Then the decoding of the $(\mathbf{s}, \mathbf{k}, t)$ DE protocol guarantees the correctness.

\subsubsection{Document exchange under edit distance}
We now describe our protocol for document exchange under edit distance, and show a connection to the problem of document exchange under Hamming distance with asymmetric information. %of the edit distance DE protocol, assuming we have protocols for DE with asymmetric information, having constant factor close to optimal communication complexity.

On a high level, our protocol follows the leveled structure used in several previous works \cite{irmak2005improved, cheng2018deterministic, haeupler2018optimal}. The protocol proceeds in $L=O(\log(\frac{n}{k}))$ levels where in each level, Alice sends a sketch of her string $x$ with $O(k)$ bits to Bob. Bob then uses all the sketches and his string $y$ to recover $x$.
%Alice generating the sketch proceeds in $L=O(\log(\frac{n}{k}))$ levels. In each level she computes some sketches about her string $x$. And her final message to Bob is the concatenation of the sketches. After receiving the message, Bob's algorithm also proceeds in $L$ levels, where in each level he uses the corresponding sketch to recover some parts of $x$. 

%Our construction inherits the structure. It is as follows.

On Alice's side, in the first level she divides her string into $\Theta(k)$ blocks where each block has size $O(\frac{n}{k})$. In each subsequent level, every block from the previous level is divided evenly into two blocks, and this ends when the block size becomes $O(\log \frac{n}{k})$, which takes $O(\log(\frac{n}{k}))$ levels. In each level, Alice applies a different random hash function to every block using the shared randomness, and computes a sketch based on the hash values. On Bob's side, his recovering process also proceeds in $L$ levels, where in each level Bob maintains a string $\tilde{x}$ which is Bob's current version of Alice's string $x$. Specifically, in each level Bob also applies the same hash functions to the blocks of $\tilde{x}$ to get the hash values, then he uses this level's sketch to recover the correct hash values of Alice's bocks. Bob will then find the blocks in $\tilde{x}$ which have inconsistent hash values with Alice's blocks, and update these blocks using his string $y$ by computing a non overlapping matching between $y$'s blocks and the corresponding hash values. An important property of the protocol is that in each level, the number of different blocks between $x$ and $\tilde{x}$ is always bounded by $O(k)$ with high probability. This ensures that Alice can send a short sketch to Bob for him to recover the correct hash values of all blocks.

To ensure that Alice's sketch in each level has length $O(k)$, there are several non trivial issues. First, every hash function needs to have only $O(1)$ bits of output, as in \cite{haeupler2018optimal}. Second, even so, the general task of recovering $s$ hash values with $O(k)$ errors needs to use a sketch of size at least $\log \binom{s}{k}=\Omega(k \log\frac{s}{k})$, where $s$ is the number of blocks in the current level. This can be as large as $\Omega(k \log\frac{n}{k})$ when $s$ becomes $n^{\Omega(1)}$, and thus will be problematic. To fix this issue, \cite{haeupler2018optimal} uses a more careful analysis called ``t-witness'' to show that in each level, the total number of possible error patterns is $2^{O(k)}$ with high probability, instead of $\binom{s}{k}$. Thus, in theory one can simply use another random hash function with $O(k)$ bits of output to distinguish all error patterns, and this brings the sketch size back to $O(k)$.\ However, simply doing this will result in an exponential running time since it involves exhaustive search. Thus,  \cite{haeupler2018optimal} needs to first randomly partition the blocks into bins, such that with high probability each bin has $O(\log n)$ hash errors. The exhaustive search in each bin now takes $\poly(n)$ time. Unfortunately, this also increases the error probability from $2^{-\Omega(k)}$ to $1/\poly(n)$. 

In our protocol, we instead replace the approach of random partitioning and exhaustive search in  \cite{haeupler2018optimal} by a direct efficient approach, thus improving the error probability to be exponentially small. We achieve this by establishing a connection to the problem of document exchange under Hamming distance with asymmetric information, as follows. 

Intuitively, in Bob's process of recovering the string $x$, in each level Bob keeps track of the positions of the possible blocks where his version $\tilde{x}$ and $x$ may be different (we call these blocks bad). More specifically, recall that we can show in each level, with high probability there are at most $O(k)$ bad blocks. In the next level the number of these blocks will at most double due to splitting, however since we use random hash functions with $O(1)$ output bits, we can show that in the next level with high probability Bob will detect $O(k)$ bad blocks and update them. Some of the updated blocks may still be bad, but Bob knows the positions of all updated blocks, and he also knows that there are at most $O(k)$ bad blocks in them after the update. Now, suppose these updates happen in level $j$, and Bob is now in level $i > j$. Then the $O(k)$ updated blocks will split into $O(k 2^{i-j})$ smaller blocks. If any of these smaller blocks is bad and it remains undetected so far, then it must have gone through $j-i$ different hash functions. If we choose all hash functions independently, then the probability that this happens is $2^{-c(i-j)}$ for some constant $c$. By choosing the number of output bits of the hash functions to be a large enough constant, we know that the expected number of smaller bad blocks that remain undetected so far is $O(k /2^{i-j})$. With a little extra effort, we can show that with high probability the number of these blocks is at most $k_{i-j}  = \max\{k/\log^3 \frac{n}{k}, 20k/2^{i-j}\}$, and Bob knows that these blocks are inside the subset $S_{i-j}$ with size $O(k 2^{i-j})$, which stems from the $O(k)$ updated blocks in level $j$. In other words, this gives a forest with the $O(k)$ updated blocks in level $j$ being the roots, and the at most $k_{i-j}$ bad blocks are among the $|S_{i-j}|=O(k 2^{i-j})$ leaves. 

Note that the bad blocks in level $i$ can come from the updated blocks in all previous levels, thus we get a vector $\svec=(S_1, \cdots, S_{i-1})$ and a vector $\knum=(k_1, \cdots, k_{i-1})$. Furthermore in this process, whenever a bad block stemming from some level $j$ gets detected and updated in a later level $j'$, this new block in level $j'$ will become a new root and all its descendents are removed from the set $S_{i-j}$ and put into the set $S_{i-j'}$. This ensures that the final subsets $(S_1, \cdots, S_{i-1})$ are disjoint. Finally, only Bob knows the sets $(S_1, \cdots, S_{i-1})$, but both parties know $(s_1=|S_1|, \cdots, s_{i-1}=|S_{i-1}|)$ and $(k_1, \cdots, k_{i-1})$.\ Thus, we have reduced the problem of sending the sketch in level $i$ to the problem of document exchange under Hamming distance with asymmetric information. %The correctness of this can be proved by a careful argument using induction, which we omit here. %ere we use a a similar argument ``t-witness'' technique in  \cite{haeupler2018optimal}

\subsubsection{Document exchange for a special setting of parameters}
We now give our protocol for document exchange with asymmetric information, in the special setting described above. Recall that we have $  s_i = O(k2^i), k_i = \max\{20k/2^{i-1}, k/\log^3 \frac{n}{k}\}, i\in [t], t = O(\log \frac{n}{k})$. One can compute $\ent(\snum, \knum)=\Theta(k)$ here, so our protocol for the general setting will result in sub-optimal communication complexity.\ We give a different protocol here, which uses just one expander graph instead of a sequence of expander graphs. 

The expander graph $G: [n] \times [d] \to [m]$ is generated by the shared randomness, with $m=O(k)$ and the following expansion property: for every $R \subseteq \cup_{i=1}^{t} S_i$ where $|R| \in [k/\log \frac{n}{k}, O(k)] $ and $\forall i \in [t], |R \cap S_i| \leq 20 k_i$, we have $|\Gamma(R)| \geq 0.9 d |R|$.\ Limiting the expansion to restricted sets rather than all sets $R$ with $|R| \in [k/\log \frac{n}{k}, O(k)] $ is the key to reduce the number of right vertices from $\Omega(k \log \frac{n}{k})$ to $O(k)$. Indeed, using a careful analysis of probabilities, we show that a random bipartite graph with constant $d$ and $m=O(k)$ satisfies this property with high probability. The main intuition is that the sequence $\{s_i, i\in [t]\}$ roughly increases exponentially, while the sequence $\{k_i, i\in [t]\}$ roughly decreases exponentially. 

Using this expander Alice sends her parity checks to Bob, and Bob again runs a belief propagation algorithm. The purpose of this phase is to reduce the total Hamming distance between $x$ and $\tilde{x}$ (Bob's current version of $x$, starting with $\tilde{x}=y$) to at most $k/\log \frac{n}{k}$. However, the belief propagation has tricky issues here, as the standard approach may flip much more than $20 k_i$ bits in $S_i$. This can result in a subset $R \subseteq [n]$ which does not have good expansion, thus ruining the whole process. To fix this, we prohibit the algorithm from flipping more than $20k_i$ bits in $S_i$ for each $i$. This is done by keeping track of the number of already flipped bits in each $S_i$, and for any $i$ if this number reaches $19k_i$, then subsequently in $S_i$ the algorithm will only flip bits that are previously flipped.

To show that this indeed works, at each step of the belief propagation, let $R \subseteq \cup_{i=1}^{t} S_i$ stand for the set of indices where $x$ and $\tilde{x}$ have different bits, and let $R'$ stand for $R$ restricted to the indices which we can flip (due to our modification). Thus $R'$ always has good expansion. Our first observation is that at any time, $|R'| \geq 0.9|R|$. This is because $R'$ is different from $R$ only if for some $S_i$, the number of bits already flipped is at least $19k_i$. However originally there are at most $k_i$ errors in $S_i$, so we have introduced at least $18k_i$ new errors. This means $\forall i, |R' \cap S_i | \geq 0.9 |R\cap S_i|$, and thus $|R'| \geq 0.9|R|$. Now let ($\mathsf{s}', \mathsf{u}')$ and $(\mathsf{s} , \mathsf{u})$ be the number of satisfied and unsatisfied checks in $\Gamma(R')$ and $\Gamma(R)$ respectively. We know $\mathsf{s}' + \mathsf{u}'  \geq 0.9 d |R'|$. 
Also, again by the fact that each satisfied check in $\Gamma(R)$ has at least two neighbors in $R$, we have $2\mathsf{s}' + \mathsf{u}' \leq 2\mathsf{s} + \mathsf{u} \leq  d|R| \leq \frac{10}{9} d|R'|$. 
From these two inequalities we can still deduce that $\mathsf{u}' \geq 0.7 d |R'|$, thus Bob can find a bit in $R'$ to flip. 

When this process stops, the Hamming distance between $x$ and $\tilde{x}$ is at most $k/\log \frac{n}{k}$. We can now use a deterministic document exchange protocol for Bob to recover $x$. The communication complexity is $O((k/\log \frac{n}{k})\log\frac{n}{k})=O(k)$. The only error probability here comes from the generation of the expander graph, which is $2^{-\Omega(k / \log \frac{n}{k} ) }$. We also show that the other errors in the protocol for edit distance is $2^{-\Theta(k/  \log^3 \frac{n}{k} ) }$. Thus the total error of the protocol for edit distance is $2^{-\Theta(k/  \log^3 \frac{n}{k} ) }$. When $k < \log^4 n$, we can switch to the protocol in \cite{haeupler2018optimal} which has error $1/\poly(n)$.
\section{Discussion and Open Problems}
In this paper we initiated a systematic study of document exchange and error correcting codes with asymmetric information. While we provided both lower bounds and upper bounds, as well as efficient randomized constructions that are close to optimal, there are still many interesting problems left. We list some below.

\begin{description}
\item [Question 1:] The most obvious open problem is to achieve optimal communication complexity (i.e., $\ent(\snum, \knum)$) for a one round randomized protocol. Two related questions are to reduce the error probability of the randomized protocol, and to study the case where the condition $\forall i, s_i \geq 2k_i$ does not hold. For example, is there a better deterministic protocol for the latter case?

\item  [Question 2:] A better understanding of the problem in the case of two sided asymmetric information. The results in this paper only study the case of two sided asymmetric information where $s^A+s^B \leq n$, i.e., the subsets from both parties can be disjoint in the worst case. What happens when $s^A+s^B > n$? In this case the subsets from both parties are guaranteed to overlap, and the situation becomes more complicated.

\item  [Question 3:] Two round deterministic protocol. We showed that for any one round deterministic protocol, the asymmetric information is not useful. However, by a result of Orlitsky \cite{Orlitsky90}, there exists a two round exponential time deterministic protocol with communication complexity $O(\ent(\snum, \knum)+\log n)$. The idea is that Bob sends a description of an appropriate hash function to Alice in the first round, and Alice sends the hash value of her string $x$ in the second round. The exponential running time comes from both the selection of hash functions and the recovering of $x$ using the hash value. It is an interesting open problem to see if we can design efficient protocols matching this bound. Our result suggests a way to approximate this: Bob sends a description of a sequence of appropriate expanders in the first round, and Alice sends the parity checks of her string $x$ in the second round. Using our algorithm, the recovering of $x$ in the second round is already efficient (in fact nearly linear time), however the first step of selecting the expanders still requires exponential time.

\item  [Question 4:] Optimal deterministic document exchange under edit distance. Our results also bring some hope to obtain an optimal deterministic document exchange protocol under edit distance. Especially, we have replaced the decoding by exhaustive search approach in \cite{haeupler2018optimal} by an efficient decoding algorithm. However, how to appropriate pick a hash function remains a problem. We also note that reducing the error probability is the first step towards a deterministic protocol, since if the error probability is small enough, then by a simple union bound there exists a non-uniform deterministic protocol that runs in polynomial time.
\end{description}

\section*{Paper Organization}
The rest of the paper is organized as follows.
In \cref{sec:prelim} we introduce some basic technical tools.
In \cref{sec:neg} we show lower bounds for asymmetric DE in the general setting.
In \cref{sec:GCAI} we give our protocol  for asymmetric DE in the general setting. 
In \cref{sec:cwas} we give our protocol  for asymmetric DE in a special setting.
In \cref{sec:DE} we give our protocol for DE under edit distance by using the protocol in the previous section. In \cref{sec:twoside} we generalize our results and give lower bounds and protocols for asymmetric DE with two sided information.

\section{Preliminaries}
\label{sec:prelim}

\subsection{Error correcting codes}

We will use the following well known parity check computation based on bipartite expander graphs. 

\begin{construction}[Expander Code  Encoding  \cite{sipser1994expander}]
\label{expandercode}
Let $\Gamma: [n] \times [d] \rightarrow [m]$ be a bipartite graph with $n$ left vertices, $m$ right vertices, left degree $d$.
The encoding of the $ \Gamma$-expander code, on input message $x\in \{0,1\}^{n}$, is computed as
$$ x\circ z,$$
where $z\in \{0,1\}^m$, $z[i] = \bigoplus_{j \in \Gamma^{-1}(i)} x[j], i\in [m]$.    

\end{construction}

\begin{definition}[\cite{GuruswamiUV09} ]
A bipartite graph with $n$ left vertices, $m$ right vertices and left degree $d$  
 is a $(k, a)$ expander if for every set of left vertices $ S \subseteq [n] $ of size $k$, we have $|\Gamma(S)| >  a k$. It is a $(\leq k_{\max} , a)$ expander if it is a $(k, a)$ expander for all $k \leq k_{\max}$.
\end{definition}
Here $\forall x\in [n]$, $\Gamma(x)$ outputs the set of all neighbours of $x$. It is also a set function which is defined accordingly. Also $\forall x\in [n], y\in [d]$, the function $ \Gamma:[n] \times [d] \rightarrow [m]$ is such that $\Gamma(x, y)$ is the $y$-th neighbour of $x$.

\begin{theorem}[\cite{GuruswamiUV09} ]
\label{expanderthm}
For all constants $ \alpha > 0$, for every $n \in \mathbb{N}$, $k_{\max} \leq n$, and $\epsilon > 0$, there exists an explicit
$(\leq k_{\max} ,(1-\epsilon)d)$ expander  with $n$ left vertices, $m$ right vertices, left degree $d = O((\log n)(\log k_{\max} )/\epsilon)^{1+ 1/ \alpha}$ and $m \leq d^2 k^{1+\alpha}_{\max}$. Here $d$ is a power of $2$.

\end{theorem}

The explicitness here means, given a left node, and an edge, the induced right node computed found in time $O(\log n + \log d)$.

\begin{theorem}[Classic belief propagation for decoding \cite{sipser1994expander}]
\label{BPthm}
Let $\Gamma: [n] \times [d] \rightarrow [m]$ be a $(\leq k, 3/4d)$ bipartite graph  with left degree $d_l$, right degree $d_r$.
 Let $y$ be an $n$-bit string whose
distance from a codeword $x$ is at most $k/2$. Then a repeated application of the
 following decoding algorithm to $y$ will return $x$ in time $O(d_l d_r m)$.

Decoding algorithm: Upon
receiving the input $n$-bit string $y$, as long as there exists a variable such that most
of its neighbouring constraints are not satisfied, flip it.
\end{theorem}

\begin{theorem}[\cite{hoholdt1998algebraic} \cite{garcia1996asymptotic} \cite{shum2001low} Systematic Algebraic Geometry Code]
\label{agcode}

There exists an explicit construction of algebraic geometry linear $  (n, m, d)_q$-code  with $ d+m \geq n- \frac{n}{\sqrt{q} - 1.1}, q = \lceil \frac{n}{d} \rceil^2 $, polynomial-time decoding when the number of errors is less than half of the distance. Here $n, q$ should be at least some fixed constants.

Moreover for every message $x\in \mathbb{F}_q^{m}$, the codeword is $  x\circ z$ for some redundancy $z \in \mathbb{F}_q^{n-m}$. In other words, the code is systematic.

%The encoding  time complexity is $\tilde{O}(n^2)$.
%
%The decoding time complexity is $\tilde{O}(n^3)$.
\end{theorem}

\subsection{Pseudorandomness}

A distribution $X$ over $\Sigma^{n}$ is $k$-wise independent if for any $k$ variables in $X$, their marginal distribution is uniform. 

\begin{theorem}
\label{kwiseg}
There exists an explicit construction of   $\kappa$-wise independence generator $g: \{0,1\}^{s} \rightarrow \{0,1\}^n$, where $s = O(  \kappa \log \frac{n}{\kappa}  )$.

\end{theorem}

\begin{proof}

Let $C^{\bot }$ be an algebraic geometry linear  $(n, m, d)_q$-code constructed by Theorem \ref{agcode}, with $d = \kappa+1$, $m \geq n-O(\kappa)$, $q = \poly(n/d) = \poly(n/\kappa)$.

Consider the dual code  $C = (C^{\bot } )^{\bot}$. By duality of codes, its message length is $ n-m = O(\kappa) $. Let the generator be $g(\cdot) = C(\cdot)$, i.e. the encoding function of $C$. Note that the seed length in bits is $ s = (n-m)\log q =   O(  \kappa \log \frac{n}{\kappa}  )$.

We claim that any $\kappa$ columns of the generating matrix $M \in \mathbb{F}_q^{m \times n}$ of $C$, are linearly independent. Since otherwise there will be a codeword in $C^{\bot }$, which has hamming distance $\leq \kappa = d-1$ from the codeword   $0$-vector.

Next we show $g(u) = uM$   is $\kappa$-wise independent, when $u$ is uniform.
For any $\kappa$ symbols in the output,  the  corresponding $\kappa$ columns of $M$ are linearly independent. So the matrix $M_K, K = \{\mbox{indices of these } \kappa \mbox{ columns}\}$, formed by these columns has rank $\kappa$. Thus there are $\kappa$ rows which are linearly independent. Hence each linear combination of these $\kappa$ rows in $M_K$ can uniquely represent one vector in the space of $\kappa$ symbols. So $(u M)_K$ is uniform.

To see this  is  an explicit construction, note that the encoding of $C^{\bot}$ is explicit. So the encoding of each $e_i \in \mathbb{F}^m_q, i\in [m]$, where $e_i$ is  $i$-th unit vector, is explicit. Thus the encoding matrix $M^{\bot}$, whose $i$-th row is $C^{\bot}(e_i)$, can be computed explicitly. The corresponding parity check matrix, which is actually $M$ the encoding matrix of its dual code $C$, can be computed explicitly using $M^{\bot}$ by standard procedures. So the construction is explicit.

\end{proof}

Random variables $X_1, X_2, \ldots, X_n \in \{0,1\}^{n}$ are $\eps$-almost $\kappa$-wise independent in max norm if 
$$\forall i_1, i_2, \ldots, i_{\kappa} \in [n], \forall x \in \{0,1\}^{\kappa}, |\Pr[ X_{i_1} \circ X_{i_2} \circ \cdots \circ X_{i_{\kappa}} = x] - 2^{-\kappa} | \leq \eps.$$

A function $g: \{0,1\}^{d} \rightarrow \{0,1\}^{n}$ is an $\eps$-almost $\kappa$-wise independence generator in max norm if $g(U) = X = X_{1} \circ \cdots X_{n}$ are $\eps$-almost $\kappa$-wise independent in max norm. Unless stated otherwise, we only consider max norm in the following context.

\begin{theorem}[$\eps$-almost $\kappa$-wise independence generator \cite{alon1992simple}]
\label{almostkwiseg}
There exists an explicit construction s.t. for every $n, \kappa \in \mathbb{N}$, $\eps > 0$, it computes an $\eps$-almost $\kappa$-wise independence generator $g: \{0,1\}^{d} \rightarrow \{0,1\}^n$, where $d = O(\log \frac{\kappa \log n }{\eps})$.

The construction is highly explicit in the sense that, $\forall i\in [n]$, the $i$-th output bit can be computed in time $\tilde{O}(  \log n + \log \frac{1}{\eps})$ given the seed and $i$. (The $\tilde{O}$ here hides some $\log \log n$, $\log \log (1/\eps)$ factors)
\end{theorem}

\begin{theorem}[General moment inequality for $k$-wise independence]
\label{kwiseconcentrbound}
Let $X_i \in \{0,1\}, i=1,\ldots, n$, be a sequence of $k$-wise independent random variables. Let $X= \sum_{i=1}^n X_i$. 

For every $\varepsilon > 0$,
$$ \Pr[ X \geq  (1+ \varepsilon)  \mathbb{E}X] \leq \left(\frac{1}{1+\varepsilon} \right)^k.  $$

\end{theorem}

\subsection{LCS and Matching}

Consider two strings $x \in  \{0,1\}^{pn}, y\in  \{0, 1\}^{n'}$, hash functions $ h_j:\{0,1\}^{p} \rightarrow \{0,1\}^q, j\in [n] $. 
A monotone matching $w = ((\rho_1, \rho'_1), \ldots, (\rho_{|w|}, \rho'_{|w|}) )$ between $x, y$ under $h_j, j\in [n]$ is s.t. for every $ i\in [|w|] $, $ h_{\rho_i} \left( x[\rho_i, \rho_i + p)\right) = h_{\rho_i} \left( y[\rho'_i, \rho'_i + p) \right)  $, where  $\rho_i \in [pn],  \rho'_i\in [n'] $.
Also we consider $x$ as being cut into length $p$ blocks and each $\rho_i$ has to be a starting position of a block in $x$. 
 
%In other words, for the $i$-th match, the  block  $x[\rho_i, \rho_i + p)$ of $x$, and the interval $ y[\rho'_i, \rho'_i + p)$ of $y$, have the same hash value under hash function $ h_{\rho_i} $.
%Also we usually consider $x$ as being cut into blocks, each having length $p$, and $\rho_i$ points to a starting position of a block, while in contrary $\rho'_i$ can point to any position of $y$.

\begin{lemma} 
\label{numOfPossibleMatchings}
For any $x \in \{0,1\}^{pn}, y\in \{0, 1\}^{n'}, k\in \mathbb{N}$, $S \subseteq [n], |S| = s$, $ h_j:\{0,1\}^{p} \rightarrow \{0,1\}^q, j\in [n] $,  the number of matchings $w  = ((\rho_1, \rho'_1), \ldots, (\rho_{|w|}, \rho'_{|w|}) )$ between $x_S$ and $y$ under $h_j, j\in [n]$ s.t.
$$    |\rho'_1 - \rho_1| + |(\rho'_2 - \rho'_1 ) - (\rho_2 -\rho_1) | + \cdots + | ( \rho'_{|w|} -\rho'_{|w|-1}) - (\rho_{|w|} -\rho_{|w|-1})| \leq  k ,    $$
is at most $  2^{2s +k(\log \frac{k+s-1}{k} +   \log e) } $.

\end{lemma}

Here $x_S$ refers to the sequence of blocks of $x$.  The $j$-th block of it is $x_S[j] \in \{0, 1\}^p, j\in [s]$.
We use $\pos(j)$ to refer to the starting position of   block $x_S[j]$ in $x$. 

\begin{proof}

Let's first consider the number of matchings with length $\tilde{s} \in \{0, 1,2,\ldots, s\}$. The number of possible $ \rho_1, \ldots, \rho_{\tilde{s}} $ is   ${s \choose \tilde{s}}$.

Assume   $|(\rho '_j - \rho '_{j-1} ) - (\rho _j -\rho _{j-1}) |  = k_j, j = 1,\ldots, \tilde{s}$,  $\rho '_0 - \rho _0  =  0 $.

For a sequence of fixed $ \rho_1, \ldots, \rho_{\tilde{s}} $,  the total number of possible matchings $w$ s.t.  $$    |\rho'_1 - \rho_1| + |(\rho'_2 - \rho'_1 ) - (\rho_2 -\rho_1) | + \cdots + | ( \rho'_{\tilde{s}} -\rho'_{\tilde{s}-1}) - (\rho_{\tilde{s}} -\rho_{\tilde{s}-1})| = \sum_{j=1}^{\tilde{s}} k_j \leq  k ,    $$
is at most 
$$2^{\tilde{s} } {k+\tilde{s}-1 \choose \tilde{s}-1}= 2^{\tilde{s} } {k+\tilde{s}-1 \choose k} \leq  2^{ \tilde{s }+k(\log \frac{k+\tilde{s}-1}{k} + \log e)} \leq  2^{ s+k(\log \frac{k+s-1}{k} + \log e)} ,$$
Since each sequence of $ \rho'_j, j\in [\tilde{s}] $ one-on-one corresponds to a sequence of $k_j \in \mathbb{N}, j\in [\tilde{s}]$ and the signs of $ (\rho '_j - \rho '_j ) - (\rho _j -\rho _j), j = 1,\ldots, \tilde{s} $.

So the overall number of possibilities is at most 
$$\sum_{\tilde{s}=0}^{s} {s \choose \tilde{s}}  2^{ s+k(\log \frac{k+s-1}{k} + \log e)}  \leq  2^{2s +k(\log \frac{k+s-1}{k} +   \log e) }.$$

\end{proof}

\begin{lemma}[DP for LCS within  $k$ edit operations]
 \label{dpformatch}

There is an algorithm, on input $x\in    \{0,1\}^{pn}, y\in \{0,1\}^{n' = O(np)}, S\subseteq [n], k =\ED(x, y) $, hash functions $h_i:\{0,1\}^p \rightarrow \{0,1\}^q, i\in [n]$,  outputs a   monotone matching $w = ((u_1, u'_1), \ldots, (u_{|w|}, u'_{|w|} ))$ between $x_S $ and $y$ under $h_i, i\in [n]$ s.t. $|w| \geq |S| - k$, and
\[  |u'_1 - u_1| + |(u'_2 - u'_1 ) - (u_2 -u_1) | + \cdots + | ( u'_{|w|} -u'_{|w|-1}) - (u_{|w|} -u_{|w|-1})| \leq  k. \]

\end{lemma}

\begin{proof}

We present a dynamic programming  to compute the maximum matching.

For every $  j\in [|S|], j'\in [n']$, let $f(j, j', l)$ be the maximum matching $w = ((u_1, u'_1), \ldots, (u_{|w|}, u'_{|w|} ))$ between $x_S[1, j]$ and $ y [1, j'] $ under $h_i, i\in [n]$, s.t.
\begin{itemize}
\item  
$ g(w) = |u'_1 - u_1| + |(u'_2 - u'_1 ) - (u_2 -u_1) | + \cdots + | ( u'_{|w|} -u'_{|w|-1}) - (u_{|w|} -u_{|w|-1})|   \leq  l $;

%\item  $u_1,  \ldots, u_{|w|} \in [pn], u'_1, \ldots, u'_{|w|}\in [n'] $;

\item The last match matches $x_S[j]$ to $ y[j', j'+p) $.
%i.e. $u_{|w|}$ is the position of $x_S[j]$ in $x$, and $u'_{|w|} = j$.
\end{itemize}
If there is no such matching, then $f(j, j', l) $ is $\mynull$   and $g(\emptyset) = -\infty$.

We compute $f(j, j', l)$ as follows.

%To initialize, if $ h_j(x[j]) =  h_j(\:y[j', j'+p)\: )$,  we let $f(j, j', |j - j'|) = \{(j, j')\}$.
To initialize,   we let $f(0, 0, 0) = \emptyset $.

For every $ j\in [ |S| ], j' \in [n'], l\leq k$,
\begin{enumerate}

\item If $ h_j(x[j]) \neq  h_j(\:y[j', j'+p)\: )$, then $f(j, j', l)  $ is $\mynull$ and $g(\mynull) = \infty$; 

\item Pick the maximum matching $w_1$ in  
$$ M =  \{f(j_1, j_1', l_1) \mid j_1 < j, j'_1 < j', l_1 \leq l, g( f(j_1, j_1', l_1) ) + |\pos(j)- \pos(j_1) - (j'-j'_1)| \leq l \}; $$

\item Let $f(j, j', l)  = w_1 \cup \{(\pos(j), j')\}.$

\end{enumerate}

Finally we use an exhaustive search to find the  maximum matching among $ f(j, j', k), j\in [n ], j'\in [n']$ and output.

Next we prove the correctness.

We first claim that,   there exists a matching $w^*$ of length $|S| -k$ between $ x_S $ and $y$ which  has $g(w^*) \leq k$. This is because we can match each $i\in S$ to exactly the same entry after the $k$ edit operations to get $w^*$. Here $g(w^*) \leq k$ is because otherwise the edit distance between $x$ and $y$ is larger than $k$.

Assuming the $ i $-th pair in $w^*$ matches $x_S[j_i]$ to $y[j'_i, j'_i +p)$.  Let $w^*_i$ be the first $i$ pairs of $w^*$. 

We use induction to show that $ |f(j_{|w^*|}, j'_{|w^*|}, g(w^*) )| \geq |w^*|$.  

For the base case, note that   $ |f(j_1, j'_1, g(w^*_1))| \geq 1  $ since at least we have a matching $f(0,0,0) \circ (\pos(j_1), j'_1)$. 

Suppose for  $ i \geq 1 $,   $ |f(j_i, j'_i, g(w^*_i))| \geq i $. For $i+1$, by our construction to compute $f(j_{i+1}, j'_{i+1}, g(w^*_{i+1}))$, we know 
$$ g(f(j_i, j'_i, g(w^*_i)) ) + |\pos(j_{i+1}) - \pos(j_{i}) - ( j'_{i+1} -j'_{i} )| \leq g(w^*_i) + | \pos(j_{i+1})  - \pos(j_{i}) - ( j'_{i+1} -j'_{i} )| = g(w^*_{i+1}).$$
So $ f(j_i, j'_i, g(w^*_i)) $ is in $M$.
Since in the second stage of the computing of $f(j_{i+1}, j'_{i+1}, g(w^*_{i+1}))$ we pick the maximum matching in $M$ and add one more match to it, we know $$|f(j_{i+1}, j'_{i+1}, g(w^*_{i+1}))| \geq |f(j_i, j'_i, g(w^*_i)|+1 \geq i+1 .$$

This shows the induction step.

As a result,   the output matching has length at least $ |w^*|  $.
\end{proof}

\section{Negative Result}
\label{sec:neg}
In this section, we show some lower bounds for the asymmetric document exchange and error correcting codes. Given the vectors $\snum=(s_1, \cdots, s_t)$ and  $\knum=(k_1, \cdots, k_t)$, we define 
\[\ent(\snum, \knum)=\log \left (\prod_{i=1}^{t} \left (\sum_{j=0}^{k_i}\binom{s_i}{j} \right ) \right )=\sum_{i=1}^{t} \log \left (\sum_{j=0}^{k_i}\binom{s_i}{j} \right ).\]

Similarly, for two integers $s$ and $k$ with $s \geq k$, we define

\[\ent(s, k)=\log \left (\sum_{j=0}^{k}\binom{s}{j} \right ) .\]

Note that in particular we have $\ent(\snum, \knum) \geq \sum_{i=1}^{t} k_i \log(s_i/k_i)$ and $\ent(s, k) \geq k \log (s/k)$.

We now have the following theorems.

\begin{theorem}\label{thm:AliceR}
In an $(\snum, \knum, t)$ asymmetric DE problem where Bob has the vector of subsets ${\cal S}=(S_1, \cdots, S_t)$, let $k=\sum_{i=1}^{t} k_i$ and suppose Alice learns Bob's string. Then any deterministic protocol has communication complexity at least $\ent(n, k)$, and any randomized protocol with success probability $\geq 1/2$ has communication complexity at least $\ent(n, k)-1$. This holds even if Alice knows $\snum$ and $\knum$.
\end{theorem}

\begin{proof}
Assume for the sake of contradiction that there is a deterministic protocol with communication complexity less than $\ent(n, k)$. Fix Alice's string $x$, and the number of strings $y$ within Hamming distance $k$ of $x$ is exactly $2^{\ent(n, k)}$. For each of these strings, one can define a vector of subsets ${\cal S}=(S_1, \cdots, S_t)$ consistent with $\snum=(s_1, \cdots, s_t)$ such that with each subset $S_i$ the Hamming distance is exactly $k_i$. Since the transcript of the protocol is a deterministic function of $(x, y, \cal S, \snum, \knum, t)$, at least two different $y$'s from Bob's side will produce the same transcript. Now since Alice's final output is a deterministic function of $x$ and the transcript, this means Alice will not be able to distinguish the two different $y$'s, contradicting that the protocol always succeeds. 

Similarly, assume for the sake of contradiction that there is a randomized protocol with communication complexity less than $\ent(n, k)-1$, that succeeds with probability $\geq 1/2$. Fix Alice's string $x$ and consider the $2^{\ent(n, k)}$ different strings $y$ as above. By an averaging argument there is a fixing of the random bits used, such that the protocol succeeds for at least $2^{\ent(n, k)-1}$ $y$'s. Since the protocol is now fixed the same argument gives a contradiction.
\end{proof}

We now consider the case where Bob tries to learn Alice's string, and we have the following theorem.

\begin{theorem}\label{thm:BobR}
In an $(\snum, \knum, t)$ asymmetric DE problem where Bob has the vector of subsets ${\cal S}=(S_1, \cdots, S_t)$, let $k=\sum_{i=1}^{t} k_i$ and suppose Bob learns Alice's string. Then any randomized protocol with success probability $\geq 1/2$ has communication complexity at least $\ent(\snum, \knum)-1$. Furthermore if $\forall i, s_i=|S_i| \geq 2 k_i$, then any one round deterministic protocol has communication complexity at least $\ent(n, k)$. This holds even if Alice knows $\snum$ and $\knum$.
\end{theorem}

\begin{proof}
Assume for the sake of contradiction that there is a randomized protocol with communication complexity less than $\ent(\snum, \knum)-1$, that succeeds with probability $\geq 1/2$. Fix Bob's string $y$, and the number of strings $x$ within Hamming distance $k_i$ in each subset $S_i$ is exactly $2^{\ent(\snum, \knum)}$. By an averaging argument there is a fixing of the random bits used, such that the protocol succeeds for at least $2^{\ent(\snum, \knum)-1}$ $x$'s. Thus, again at least two different $x$'s will produce the same transcript, and Bob will not be able to distinguish. This gives a contradiction.

Similarly, assume for the sake of contradiction that there is a deterministic protocol with communication complexity less than $\ent(n, k)$. This means two different $x$'s will produce the same transcript in a one-round protocol, where the transcript is a deterministic function of $(x, \snum, \knum, t)$. For these two different $x$'s, as long as $\forall i, s_i=|S_i| \geq 2 k_i$, one can define a vector of subsets ${\cal S}=(S_1, \cdots, S_t)$ such that for each $x$, the Hamming distance between the corresponding substrings of $x$ and $y$ in $S_i$ is exactly $k_i$. Thus the inputs to Bob are the same for the two $x$'s. Since Bob's final output is a deterministic function of his inputs and the transcript, Bob will not be able to distinguish the two different $x$'s, a contradiction. 
\end{proof}

We also have the following theorem for asymmetric error correcting codes.

\begin{theorem}\label{thm:ECC}
In an $(\snum, \knum, t)$ asymmetric ECC problem where Bob has the vector of subsets ${\cal S}=(S_1, \cdots, S_t)$, let $k=\sum_{i=1}^{t} k_i$. If $\forall i, s_i=|S_i| \geq 2 k_i$, then any deterministic code must have distance at least $2k+1$. In particular, $m \leq n-\ent(n, k)$. Furthermore, any randomized code with success probability $\geq 1/2$ must have message length $m \leq n-\ent(\snum, \knum)+1$. 
\end{theorem}

\begin{proof}
Assume for the sake of contradiction that there is a deterministic code with distance at most $2k$. This means there are two different codewords $\Enc(x_1)$ and $\Enc(x_2)$ with Hamming distance at most $2k$. Thus, an adversary can come up with two error strings $z_1, z_2$ where each $z_j$ has exactly $k$ $1$'s, such that $\Enc(x_1) \oplus z_1=\Enc(x_2) \oplus z_2=y$. As long as $\forall i, s_i=|S_i| \geq 2 k_i$, one can define a vector of subsets ${\cal S}=(S_1, \cdots, S_t)$ such that for each $z_j$, the number of $1$'s in the subset $S_i$ is exactly $k_i$. Thus for $x_1$ and $x_2$, Bob receives the same string $y$ and his other inputs are also the same. This means that Bob will not be able to distinguish $x_i$ and $x_j$, a contradiction. 

Now assume for the sake of contradiction that there is a randomized code with success probability $\geq 1/2$ which has message length $m > n-\ent(\snum, \knum)+1$. By an averaging argument there exists a fixing of the random bits used in encoding and decoding, that succeeds for $2^{m-1} > 2^{n-\ent(\snum, \knum)}$ messages. Note that for any codeword, the number of all strings which have Hamming distance at most $k_i$ in the subset $S_i$ to the codeword is $2^{\ent(\snum, \knum)}$. This implies that there exists two different codewords $\Enc(x_1)$ and $\Enc(x_2)$ and a string $y$ such that for each $\Enc(x_j)$, $y$ has Hamming distance at most $k_i$ in the subset $S_i$ to the codeword $\Enc(x_j)$. An adversary can thus change $\Enc(x_1)$ and $\Enc(x_2)$ into the same string $y$, and both error patterns are consistent with $(\cal S, \snum, \knum)$. Thus Bob will not be able to distinguish $x_i$ and $x_j$, a contradiction.  
\end{proof}

\section{Document Exchange and Error Correcting Codes with Asymmetric Information in the General Setting}
\label{sec:GCAI}

We give a random protocol  for the general setting s.t. the communication complexity is close to optimal.

\subsection{Key components}

\begin{lemma}
\label{lem:SkRandom}
For every $S \subseteq [n]$, integer $k_0 \leq k \leq s =|S|$, the probability that
a random bipartite graph with $n$ left vertices, $m \geq 2 dk 2^{1/\delta}  $ right vertices,  left degree $d = O(\log \frac{2s}{k})  $, having  
\begin{equation}
\begin{split}
&\mbox{ for every } R \subseteq S, \mbox{ with } |R| \in [ k_0 , k]\\
&|\Gamma(R)| > (1-\delta) d|R|,
\end{split}
\end{equation}  
is at least $1- \eps$, where $\eps =   2^{-\Theta(\delta (\log \frac{2s}{k}) k_0 \log \frac{2k}{k_0})} $.

Note that when  $k_0 = 1$,  we get an $(n, m, d, S, \leq k, 1-\delta)$ expander with probability at least $ 1-   2^{-\Theta(\delta \log\frac{2s}{k} \log (2k)  )} \leq 1-1/\poly(s)$.
\end{lemma}

We also denote a bipartite graph with the expansion property stated as an $(n, m, d, S, [k_0, k], 1-\delta)$ expander.

\begin{proof}

The total number of sets $R$ with size $r$ is at most $ (\frac{es}{r})^r $.

For a fixed set $R$, a fixed set $T \subseteq [m], |T| =(1-\delta) d |R|$
\begin{equation} 
\Pr[ \Gamma(R) \subseteq T ] = \left(  \frac{|T|}{m}\right)^{dr} = \left( \frac{(1-\delta)dr}{m} \right)^{dr}.
\end{equation}

There are at most 
\begin{equation}
{m \choose |T|}  \leq \left( \frac{em}{|T|} \right)^{|T|} = \left( \frac{em}{(1-\delta)dr}  \right)^{(1-\delta) dr}
\end{equation}
such set $T$.

So by a union bound, 
the probability that for every $R, |R| = r$, $\Gamma(R) \leq (1-\delta) dr$ is at most 
\begin{equation}
\begin{split}
&  \left( \frac{em}{(1-\delta)dr}  \right)^{(1-\delta) dr} \times  \left( \frac{(1-\delta)dr}{m} \right)^{dr} \times     (\frac{es}{r})^r  \\
 = & e^{(1-\delta)dr} \left( \frac{(1-\delta)dr}{m} \right)^{\delta dr} \times (\frac{es}{r})^r  \\
\leq  &  e^{dr} e^{-  \delta dr \log \frac{m}{dr} }  (\frac{es}{r})^r  \\
\leq & 2^{-\Theta(\delta d r\log \frac{2k}{r})}.\\
\end{split}
\end{equation} 
by letting $m = 2dk 2^{1/\delta}$, $d = O(\log \frac{2s}{k})$.

By another union bound the probability that for every $R, |R| \in [k_0, k]$, it does not have a good expansion is at most $ \sum_{j = k_0}^k   2^{-\Theta(\delta d j \log \frac{2k}{j})}  \leq  (k-k_0+1) 2^{-\Theta(\delta d k_0 \log \frac{2k}{k_0})} \leq 2^{-\Theta(\delta d k_0 \log \frac{2k}{k_0})} $.

When $k_0 = 1$, this is at most $  2^{-\Theta(\delta \log \frac{2s}{k} \log (2k)  )} \leq 2^{-\Theta(\delta \log \frac{2s}{k} \log (2k)  )} \leq 1/\poly(s)$.

\end{proof}

\begin{lemma}
\label{lem:oneSetDec}
Assume $\Gamma$ is an $(n, m, d, S,  [k_1', 2k_1], 0.9)$ expander. Let $y$ be the expander-code encoding of $x$
using $\Gamma$. Then there is an explicit decoding which, on input $x'$ which has $k_i, i\in [t]$ errors in $S_i$ from $x$, with $k_1 \geq  k'_1 \geq c \sum_{i=2}^t k_i $, $c=10$, outputs $ \tilde{x} $ that  has at most $k_1'$ errors in $S_1$.

\end{lemma}

\begin{proof}

We propose the following algorithm. For every iteration, find the first bit in $S_1$  s.t.  it has more unsatisfied checks than  satisfied ones.
Loop until we cannot find such bit anymore.

Now we show this works.
Assume there are at least $k'_1$ errors in  $S_1$. Denote $A$ as the set of indices of these errors. 
Let $s$ be the number of satisfied neighbors of $A_1= A\cap S_1 $. 
Let $u$ be the number of unsatisfied neighbors of $A_1  $.
By the expander property, $|\Gamma(A)| \geq  0.9d |A_1|$.  
So
\begin{align}
\label{eq:onesetdeceq1}
s + u \geq 0.9 d |A_1|.
\end{align} 
On the other hand, each satisfied check is connected to at least one vertex in $A_1$ since it is in $\Gamma(A_1)$. Thus it has to be connected to at least $2$ vertices in $A$ to make it to be satisfied.
Also each unsatisfied check is connected to at least $1$ vertex in $A_1$.
Hence 
\begin{align}
\label{eq:onesetdeceq2}
2s + u \leq   d |A| \leq d |A_1| + d \sum_{i=2}^t k_i \leq  (1+\frac{1}{c})d |A_1|.
\end{align} 
By \cref{eq:onesetdeceq1} and \cref{eq:onesetdeceq2}, 
$$u \geq 0.8d|A_1|.$$
So there has to be $\geq 0.1$ fraction of vertices in $S_1$ having more unsatisfied checks than satisfied ones.
As a result, the algorithm can find a bit to flip and $u$ is decreasing.
On the other hand, if at some iteration, $|A_1| = 2k_1$, then $u \geq 1.6 dk_1$ but initially $u\leq d k_1 $ which contradicts that $u$ is decreasing.
As a result, the iterations will continue until there are less than $k'_1$ errors in $S_1$. 

\end{proof}

\begin{theorem}
\label{thm:onesetprotocol}

There is an efficient $1$-round protocol s.t. for  every $(s, k)$ DE problem,   it has  communication complexity $O(k \log \frac{2s}{k})$, success probability $1-      2^{-\Theta( \log \frac{2s}{k} \log k  )}$.

\end{theorem}

\begin{proof}

We first generate a random bipartite graph with $n $ left vertices, left degree $d = O(\log \frac{2s}{k})$,  $m = O(d k )$ right vertices. By \cref{lem:SkRandom}, with probability  $1-     2^{-\Theta( \log \frac{2s}{k} \log k  )}$, it is an $(n, m, d, S, \leq 2k, 0.9)$ expander $\Gamma$. We use $\Gamma$ to compute the   sketch $z$ of $x$. 

To decode, we use $y$, $z$ and $\Gamma$. By \cref{lem:oneSetDec}, we can get $x$ correctly.

The running time of both parties are $\tilde{O}(n)$.

\end{proof}

\subsection{ The protocol }

Without loss of generality, we assume $k_1 \geq k_2 \geq \cdots \geq k_t$.

\begin{theorem}
 
\label{thm:GCAI1}
There is a $1$-way efficient protocol s.t. for every $( \mathbf{s}, \mathbf{k}, t  )$ DE with $k_i \leq s_i/2, \forall i\in [t]$,    it has success probability  $  1-2^{- \Omega( k_t) } - 1/\poly(s)$, communication complexity $O(t^2 \sum_{i\in [t]} k_i \log \frac{s_i}{k_i})$.

\end{theorem}

\begin{construction}
Efficient protocol for  $( \mathbf{s}, \mathbf{k}, t  )$ DE .

Alice: on input $x$,

\begin{enumerate}[label*=\arabic*.]

\item Let $i' = 0, k' = 0$, string $z  $ be  empty string;

\begin{enumerate}[label*=\arabic*.]

\item While $i' \leq t-2$,  find $i > i'$ s.t. $k' + \sum_{j = i'+1}^{i} k_j > k''$, where $ k'' = c \sum_{j = i +1}^t k_j$; If cannot find $i$ then break the iterations;

%\item Generate an $(n, m, d, \cup_{j=1}^{i} S_j, [k'', 2(k' + \sum_{j = i'+1}^{i} k_j) ], 0.9)$-expander $\Gamma$ by \cref{lem:SkRandom};

\item
Generate an $(n, m, d, \cup_{j=1}^{i} S_j, [k'', 2(k' + \sum_{j = i'+1}^{i} k_j )], 0.9)$-expander $\Gamma$ by \cref{lem:SkRandom}, where $d  = O(\log \frac{ \sum_{j=1}^i s_j  }{k'+ \sum_{j=i'+1}^{i} k_j }  )$;

%\item Generate an $(n, m, d)$ random bipartite expander $\Gamma$, where $d  = O(\log \frac{ \sum_{j=1}^i s_j  }{k'+ \sum_{j=i'+1}^{i} k_j }  )$;

\item Compute $z_{i}$ which is the expander code of $x$ using $\Gamma$, $z = z \circ z_{i}$;

\item Let $ i' = i, k' = k'' $.

\end{enumerate} 

\item Encode $x$ to be $z_{\mathrm{final}}$ by using a $(n, m, d = O(\log \frac{s}{k_{\mathrm{final}}}), S, \leq 2 k_\mathrm{final}, 0.9   )$ expander $\Gamma_{\mathrm{final}}$ generated by  \cref{lem:SkRandom}, where $k_\mathrm{final} = k'+ \sum_{j = i'+1}^t k_j $;  

%\item Encode $x$ to be $z_{\mathrm{final}}$ by using a $(n, m, d_{\mathrm{final}} = O(\log \frac{ s}{k_{\mathrm{final}}}) )$ random bipartite expander $\Gamma_{\mathrm{final}}$  where $k_\mathrm{final} = k'+ \sum_{j = i'+1}^t k_j $;  

\item Send $z \circ z_{\mathrm{final}}$ to Bob.

\end{enumerate} 

Bob: on input $y$, $\mathbf{S}, \mathbf{k}$, together with the message $z \circ z_{\mathrm{final}}$ from Alice;

\begin{enumerate}[label*=\arabic*.]

\item Let $i' = 0, k' = 0, y' = y$;

\begin{enumerate}[label*=\arabic*.]

\item While $i' \neq t$,  find $i > i'$ s.t. $k' + \sum_{j = i'+1}^{i} k_j > k''$, where $ k'' = c \sum_{j = i+1}^t k_j, c= 10$;

\item Generate an $(n, m, d, \cup_{j=1}^{i} S_j, [k'', 2(k' + \sum_{j = i'+1}^{i} k_j )], 0.9)$-expander $\Gamma$ by \cref{lem:SkRandom} using the same randomness as of Alice;

\item Use $\Gamma$, $z_i$ to reduce the number of errors of $y$ in $\cup_{j=1}^{i} S_j$ to be at most $k''$ by \cref{lem:oneSetDec};

\item Let $ i' = i, k' = k'' $.

\end{enumerate}

\item Decode $x$ by  \cref{lem:oneSetDec} for the $(S, k'+k_t)$ setting, using $y'$, $z_{\mathrm{final}}$,  and the expander generated     the same as the $\Gamma_{\mathrm{final} } $ of Alice;
\end{enumerate}

\end{construction}

\begin{lemma}
\label{lem:GCAI1cc}
The communication complexity is $O\left( t^2 \sum_{j\in [t]}k_j \log \frac{s_j}{k_j}   \right)$.

\end{lemma}

\begin{proof}

By \cref{lem:SkRandom}, $m$ of $\Gamma$ is $ O\left( (k' + \sum_{j = i'+1}^{i} k_j) \log \frac{\sum_{j=1}^{i} s_j }{ k' + \sum_{j = i'+1}^{i} k_j  }    \right)  $.

Note that in the first iteration, the algorithm may pick a $i\in [t]$. But in the succeeding iterations, it will always take $i = i'+1$, since $k' + k_{i'+1} > k''$ and we always assume $k_{i' + 1 } >0$.

For the first iteration, we have 
\begin{align*}
m  
& = O\left( (\sum_{j = 1}^{i} k_j) \log \frac{\sum_{j=1}^{i} s_j }{   \sum_{j = 1}^{i} k_j  }   \right)\\
& \leq O\left(   (c \sum_{j = i}^{t} k_j + k_i) \log   \frac{\sum_{j=1}^{i} s_j }{   \sum_{j = 1}^{i} k_j  }        \right)      && \text{Since }  \sum_{j=1}^{i-1} k_j \leq c \sum_{j=i}^{t}k_j    \\
& \leq O\left(   (c \sum_{j = i}^{t} k_j + k_i) \log   \frac{ \sum_{j=1}^{i} s_j }{ \frac{1}{2} \sum_{j = 1}^{i} k_j   + \frac{1}{2}k_i     }        \right)      && \text{Decreasing the denominator }  \\
& \leq O\left(   (c \sum_{j = i}^{t} k_j + k_i) \log   \frac{ \sum_{j=1}^{i} s_j }{ \frac{1}{2}( c\sum_{j = i+1}^{t} k_j   +  k_i )    }        \right)      && \text{Since }  \sum_{j=1}^{i } k_j > c \sum_{j=i+1}^{t}k_j   \\
& \leq O\left(   (c \sum_{j = i}^{t} k_j + k_i) \log   \frac{ \sum_{j=1}^{i} s_j }{c\sum_{j = i+1}^{t} k_j   +  k_i   }         \right) && \text{Because of big-O notation} \\
& \leq O\left(   \overline{k} \log   \frac{ \sum_{j=1}^{i} s_j }{ \overline{k}  }        \right)  && \text{Let } \overline{k} =    (c \sum_{j = i}^{t} k_j + k_i)\\
& \leq  O\left(   \overline{k} \log \prod_{j=1}^i ( \frac{ s_j }{ \overline{k}  } +1)       \right)   && \log(\cdot) \text{ is an increasing function} \\
& = O\left(\sum_{j=1}^i \overline{k} \log (\frac{s_j}{\overline{k}}  +1)  \right).
\end{align*}

For each $j\in [i]$, if $ \overline{k} > k_j $, then since $\overline{k} \leq (c+1)tk_j$,
$$ \overline{k} \log (\frac{s_j}{\overline{k}}  +1) \leq (c+1)tk_j \log (\frac{s_j}{k_j}+1) \leq 2(c+1)t k_j \log \frac{s_j}{k_j};$$
Otherwise if $\overline{k} \leq k_j$, then since $k_j \leq \frac{1}{2}s_j$,
$$ \overline{k} \log (\frac{s_j}{\overline{k}}  +1) \leq  \overline{k}   \log  \frac{2s_j}{\overline{k}}    \leq O( k_j \log \frac{s_j}{k_j}).  $$

Hence $m  = O\left( t \sum_{j=1}^i k_j \log \frac{s_j}{k_j}    \right)$.

Next we consider the cases where we are in iterations from the $2$nd to the last.
We have
\begin{align*}
m  
& = O\left( (k'+ k_i) \log \frac{\sum_{j=1}^{i} s_j }{  k'+ k_i  }   \right)\\
& = O\left(  \overline{k} \log   \frac{\sum_{j=1}^{i} s_j }{  \overline{k} }        \right)       && \text{Let } \overline{k} =  k'+ k_i =  c \sum_{j=i}^{t}k_j + k_i   \\
& \leq  O\left(   \overline{k} \log \prod_{j=1}^i ( \frac{ s_j }{ \overline{k}  } +1)       \right)   && \log(\cdot) \text{ is an increasing function} \\
& = O\left(\sum_{j=1}^i \overline{k} \log (\frac{s_j}{\overline{k}}  +1)  \right).
\end{align*}

For each $j\in [i]$, if $ \overline{k} > k_j $, then again since $\overline{k} \leq (c+1)tk_j$,
$$ \overline{k} \log (\frac{s_j}{\overline{k}}  +1) \leq (c+1)tk_j \log (\frac{s_j}{k_j}+1) \leq 2(c+1)t k_j \log \frac{s_j}{k_j};$$
Otherwise if $\overline{k} \leq k_j$, then since $k_j \leq \frac{1}{2}s_j$,
$$ \overline{k} \log (\frac{s_j}{\overline{k}}  +1) \leq \overline{k}   \log  \frac{2s_j}{\overline{k}}  \leq  O( k_j \log \frac{s_j}{k_j}).  $$

Hence $m  = O\left( t \sum_{j=1}^i k_j \log \frac{s_j}{k_j}    \right)$.

As there are at most $t$ iterations, the total communication complexity is $tm = O\left(t^2 \sum_{j=1}^t k_j \log \frac{s_j}{k_j}    \right) $.

\end{proof}

Next we show the correctness.

\begin{lemma}
\label{lem:GCAI1correct}
Bob can compute $x$ correctly with probability at least $  1-2^{- \Omega( k_t) } - 1/\poly(s)$.

\end{lemma}

\begin{proof}

In the first iteration, since $\Gamma$ is an $(n, m, d, \cup_{j=1}^i S_j, [c\sum_{j=i+1}^t k_j, 2(\sum_{j=1}^i k_j)] )$ expander, by \cref{lem:oneSetDec}, we can successfully reduce the number of errors in $\cup_{j=1}^i S_j$ to be $\leq k''$.

Note that as long as $k_{i'+1} > 0$, the number $i$, found in the iteration, will be $i'+1$. So the iteration will continue until $i' = t-1$. After the iterations, the number of errors in $S$ is at most $ k'+k_t = (c+1)k_t$.

Finally,  using $z_{\mathrm{final}}$ and  $\Gamma_{\mathrm{final}}$,  by \cref{lem:oneSetDec}, Bob can compute $x$ correctly.

The protocol succeeds once all random expander graphs are as desired. For random expander graph in iteration $i$, the success probability is $1-   2^{-\Omega(d k'' \log \frac{2k'}{k''})  )} \leq 1- 2^{-\Omega( d k'' )}$, by \cref{lem:SkRandom}. So by a union bound, the  probability, that all  iterations success, is at least $1- 2^{-\Omega(k_t)}$.
In the final step, the success probability is $1- \frac{1}{\poly(s)}$ by \cref{thm:onesetprotocol}. Hence the final success probability is as desired.

\end{proof}

\begin{proof}[Proof of \cref{thm:GCAI1}]

The correctness and communication complexity immediately follows from \cref{lem:GCAI1cc}, \cref{lem:GCAI1correct}.

For the efficiency, note that in Alice's algorithm, she just randomly generate a bipartite graph with logarithmic degree. And apply the expander encoding to get the sketch. So this is in near linear time.
For Bob's algorithm, as $S_i , i\in [t]$ are disjoint, and the belief propagation can be done in near linear time. Other operations are also in near linear time. So Bob's algorithm is also in near linear time.

\end{proof}

When $t$ is large, we can group some sets together to reduce $t$ and hence get the following theorem.

\begin{theorem}
\label{thm:GCAIchi}
There is a $1$-way efficient protocol s.t.  for every $(  \mathbf{s}, \mathbf{k}, t)$ DE with $k_i \leq s_i/2, \forall i\in [t]$, it has success probability $1 - 2^{-\Omega(k_t) } - 1/\poly(s)$, communication complexity $O\left( \chi^2(\mathbf{s}, \mathbf{k}, t) \sum_{i\in [t]} k_i \log \frac{s_i}{k_i} \right)$. 

The running time of both parties are $\tilde{O}(n)$.
\end{theorem}

\begin{proof}

We cut the interval $[2, n+1)$ into $t' = O(\log \log n)$ intervals  s.t. the $j$-th interval $I_j$ is $[2^{10^{j-1}}, 2^{10^{j}} )$.  Then for all $i$ s.t. $s_i/k_i \in I_j$,  we union them to be a set $S'_j$. Also we take $k'_j$ to be the summation of the corresponding $k_i$'s. We neglect these intervals which do not cover any $s_j/k_j$, getting a new problem i.e. a $(\mathbf{s}'', \mathbf{k}'', \chi)$ error correction problem. 

By \cref{thm:GCAI1}, the communication complexity is $O\left( \chi^2 \sum_{j\in [\chi]} k''_j \log \frac{s''_j}{k''_j} \right)$. Since  $\forall j\in [ \chi ], i\in I_j$, $\log \frac{s_i}{k_i} = O(\log \frac{s''_j}{k''_j})  $, the communication complexity is actually $O\left(\chi^2 \sum_{i\in [t]} k_i \log \frac{s_i}{k_i}\right)  $.

The time complexity and success probability is implied by \cref{thm:GCAI1}.

\end{proof}

Notice that $\chi$ can only be as large as $O(\log \log n)$. So we have the following corollary.

\begin{corollary}
\label{thm:GCAI2}
There is a $1$-way efficient protocol s.t.  for every $(  \mathbf{s}, \mathbf{k}, t)$ DE with $k_i \leq s_i/2, \forall i\in [t]$, it has success probability $1 - 2^{-\Omega(k_t) } - 1/\poly(s)$, communication complexity $O\left( \log^2 \log n \sum_{i\in [t]} k_i \log \frac{s_i}{k_i} \right)$. 

The running time of both parties are $\tilde{O}(n)$.
\end{corollary}

\subsection{From DE to stochastic coding}

We show that our construction for DE can be modified to work for   stochastic coding setting.

\begin{theorem}

There is an efficient stochastic ECC   s.t. for every $(\mathbf{s}, \mathbf{k}, t)$ type errors with $k_i \leq s_i/2, \forall i\in [t]$, it has success probability $1-2^{-\Omega(k_t)} - 1/\poly(s)$, message length $n-O\left( \chi^2(\mathbf{s}, \mathbf{k}, t) \mathsf{H}(\mathbf{s}, \mathbf{k}) \right)$.

The running time of both encoding and decoding are  $\tilde{O}(n)$.
\end{theorem}

\begin{proof}

For encoding, we first compute the length of the redundancy. By the Alice's algorithm of \cref{thm:GCAIchi}, the sketch length for $(\mathbf{s}, \mathbf{k}, t)$ document exchange, on input strings of length $n$, is  $ \mathsf{s} =  O\left( \chi^2(\mathbf{s}, \mathbf{k}, t) \sum_{i\in [t]} k_i \log \frac{s_i}{k_i} \right) $. If we apply an an asymptotically good ECC $C_0$, e.g. expander codes \cite{sipser1994expander} \cite{spielman1996linear}, to encode the sketch, then the output has length $r = O(\mathsf{s})$.
Let the message length be $n-r$.

The encoding of message $x$ has two parts. The first part is the message.
The second part is the sketch for $(\mathbf{s}, \mathbf{k}, t)$ document exchange on input $ x \circ \mathbf{0} $, where $\mathbf{0}$ is an all $0$ string of length $r$. We know the sketch length is $\mathsf{s}$.
Next we apply $C_0$ on $s$ to get $z$ which has length $r$. 
The final codeword is $x\circ z$.

We claim this code can indeed resist $(\mathbf{s}, \mathbf{k}, t)$  type errors by describing the decoding along with its analysis.

For decoding, assume the input is $x'\circ z'$. 
Note that even if all errors happen on $z$,  we can decode to recover $z$ from $z'$, since $z$ is a codeword of an ECC correcting $k$ errors. 
After we get $z$. We apply Bob's algorithm of \cref{thm:GCAIchi} on $x'\circ \mathbf{0}$, using the sketch $z$. 
The decoding will success because the error type is still $(\mathbf{s}, \mathbf{k}, t)$, as we only remove some errors happened on $z$.

The success probability only comes from the success probability of \cref{thm:GCAIchi}, since that's the only part we use randomness. So the success probability is as desired.
The encoding and decoding are in near linear time since the protocol and the asymptotically good code \cite{spielman1996linear} are both in near linear time.

\end{proof}

\section{Document Exchange with Asymmetric Information in a Special Setting}
\label{sec:cwas}

We first develop a randomized two-party (Alice and Bob) one-way hamming error document exchange protocol in which Bob knows  the errors can only happen in some subsets of all positions, where in each subset the number of errors is also bounded.

The reason we consider this kind of encoding/decoding for special error patterns is that it can have shorter redundancy   than the general coding for bounded number of hamming errors.

The encoding utilize a randomized bipartite expander graph with a large expansion.

\begin{lemma}
\label{lem:expanderGeneration}
For every $n, k, k', k'', r, d, t\in \mathbb{N}$, $  k' \leq r\leq {k} \leq n$, $  k'' t \log \frac{e k 2^t}{ k'' }  \leq k' \log \frac{2k}{k'}$, $\delta \in (0, 1)$, $d \geq \delta^{-1}$, constant $c > 0$, disjoint sets $   S_i \subseteq [n], i\in [t], |S_i|   =  k 2^{O(i)}$,  $\overline{k}_i = \max( k/2^{O(i)}, k'' ) \leq |S_i|/2 $, the probability that
a random bipartite graph with $n$ left vertices, $m \geq 2 dk 2^{1/\delta}  $ right vertices,  left degree $d$, having that
\begin{align*}
&\mbox{ for every } R \subseteq \cup_{i\in [t]} S_i, |R| = r \geq k' , 
\mbox{ with }  |R\cap S_i|  \leq \overline{k}_i,  \forall   i\in [t], \\
&\mbox{ it holds } |\Gamma(R)| > (1-\delta) dr,
\end{align*}  
is at least $1- \eps$, where $\eps = 2^{-\Theta(\delta d k'\log \frac{2k}{k'})} $ .

\end{lemma}

We denote the generated expander graph as a $(n, m, d, \mathbf{S},  \mathbf{\overline{k}}, [k', k], 1-\delta)$-expander, where $\mathbf{\overline{k}}$ is the sequence of all $\overline{k}_i, i\in [t]$.

\begin{proof}

We show that a uniformly sampled bipartite graph works. The bipartite graph with $n$ left vertices, $m$ right vertices, left degree $d$, is generated as follows. Each edge, from one vertex of the left, has its ending vertex being uniformly chosen from the right vertices.  

For a fixed $R$, if $|\Gamma(R)| \leq (1-\delta) dr  $, then there exists a set $T \subseteq [m]$ s.t. $ |T| =  (1-\delta) dr, |\Gamma(R)| \subseteq T$. 
There are at most 
\begin{equation}
{m \choose |T|}  \leq \left( \frac{em}{|T|} \right)^{|T|} = \left( \frac{em}{(1-\delta)dr}  \right)^{(1-\delta) dr}
\end{equation}
such set $T$. 
For each $T$, 
\begin{equation}\label{ieq:expanderGenerationProbFactor}
\Pr[ \Gamma(R) \subseteq T ] = \left(  \frac{|T|}{m}\right)^{dr} = \left( \frac{(1-\delta)dr}{m} \right)^{dr}.
\end{equation}

Consider a fixed $r$. Assuming $ r \in [k_{j+1}, k_{j}]$, for some $j\in [t]$. 
Notice that $j  = \Theta(\log \frac{2k}{r})$. 

Let $r_i = R \cap S_i$. The total number of different sequences $r_1, \ldots, r_t$ is at most 
\begin{equation}
{r+ t \choose r}   \leq \left(  \frac{e(r+t)}{r} \right)^{r} \leq \left(  O(\frac{2k}{r} )\right)^{r} \leq 2^{O(r\log \frac{2k}{r})}.
\end{equation}

Consider a fixed sequence $r_i, i\in [t]$ with $r_i \leq \overline{k}_i$.
The total number of possibilities of $R \cap   S_j, \ldots ,R\cap S_t$ is at most 
\begin{align*}
 & \prod_{i=j}^t {|S_i| \choose r_i} \leq \prod_{i=j}^{t} {|S_i| \choose \overline{k}_i } \leq \prod_{i=j}^{t} \left( \frac{e |S_i|}{\overline{k}_i} \right)^{\overline{k}_i}  \leq \prod_{i=j}^{t'} 2^{O(i \frac{k}{2^{O(i)}} )} \cdot \prod_{i=t'}^{t} \left( \frac{e|S_t|}{ k'' } \right)^{k''} \\
 & =  2^{O\left( \sum _{i=j}^{t'} i \frac{k}{2^{O(i)}} \right)} \cdot 2^{O((t-t')k'' \log \frac{ek2^t}{k''} ) }  \leq 2^{O( \frac{k}{2^{O(j)}} j )} \cdot 2^{O( r \log \frac{2k}{r})}  = 2^{O( r \log \frac{2k}{r})} .
\end{align*}
Here $t'$ is the first index s.t. $ \overline{k}_i = k'' $.

On the other hand, the total number of possibilities of $ R \cap   S_1, \ldots, R\cap S_{j}$ is at most 
\begin{equation}
\begin{split}
\prod_{i=1}^{j} {|S_i| \choose r_i} \leq {\sum_{i=1}^j |S_i| \choose \sum_{i=1}^j r_i} \leq { O(k 2^{O(j)}) \choose \sum_{i=1}^j r_i} \leq {O(k 2^{O(j)})  \choose r} \leq \left( O(\frac{k 2^{O(j)}}{r} ) \right)^{r} \leq 2^{O(r \log \frac{2k}{r})}.
\end{split}
\end{equation}

So by a union bound, 
the probability that for every $R, |R| = r, |R\cap S_i| \leq \overline{k}_i$, $\Gamma(R) \leq (1-\delta) dr$ is at most 
\begin{equation}
\begin{split}
& \left( \frac{em}{(1-\delta)dr}  \right)^{(1-\delta) dr} \times  \left( \frac{(1-\delta)dr}{m} \right)^{dr} \times  2^{O( r \log \frac{2k}{r})} \\
 = & e^{(1-\delta)dr} \left( \frac{(1-\delta)dr}{m} \right)^{\delta dr} \times 2^{ O( r \log \frac{2k}{r} )} \\
\leq  &  e^{dr} e^{-  \delta dr \log \frac{m}{dr} } 2^{O(r \log \frac{2k}{r})}  \\
\leq & 2^{-\Theta(\delta d r\log \frac{2k}{r})}\\
\end{split}
\end{equation} 
by letting $m = 2dk 2^{1/\delta}$.

Since $k\geq r \geq k'$, it holds that $2^{-\Theta(\delta d r\log \frac{2k}{r})} \leq   2^{-\Theta(\delta d k'\log \frac{2k}{k'})} $.

\end{proof}

\begin{remark}

Note that we can use a  $\kappa = O(k d)$-wise independence generator to generate the edges of the graph. Each edge is chosen according to a random variable in a sequence that is $\kappa$-wise independent. Each random variable has support size $m$. Hence inequality (\ref{ieq:expanderGenerationProbFactor})  still holds. So we can apply the same argument. 

\end{remark}

The decoding algorithm has two   parts. Both parts use belief propagation techniques. In the first part, we reduce the number of errors slightly by using $z_1$. In the second part, we further reduce the number of errors to $0$ by using $z_2$. 

\begin{construction}[Protocol for a specific setting of parameters]
\label{constr:enhancedBPDec}
Let $n, m, d, t \in \mathbb{N}$, $k_i \in \mathbb{N}, k_i \leq n, i\in [t]$, $k' = O(k/\log \frac{n}{k}) $,   disjoint sets $S_i \subseteq [n], i\in [t]$. Let  $S = \cup_{i\in [t]} S_i$.

Let expander graph $\Gamma_1: [n] \times [d_1] \rightarrow [m_1]$, s.t.
$$\forall R \subseteq \cup_{i\in [t]} S_i \mbox{ with } |R| \in [ k', O(k)]  \mbox{ and } \forall i\in [t], |R\cap S_i|  \leq 20k_i, \mbox{ it holds } \Gamma_1(R) > 0.9 d|R|.$$

Let $C_0  $ be a systematic Algebraic Geometry code from \cref{agcode}, with alphabet $\mathbb{F}_q$, message length $  n/q  $, redundancy length $ O(k') $ correcting $2k'$ errors.

Let $x\in \{0,1\}^n$ be the original message.

The decoding takes an input string $y\in \{0,1\}^n$, parity checks $z_1 $ generated by expander encoding of $x$  using $\Gamma_1 $, and $z_2$ which is the redundancy part of $C_0(x)$.

Stage 1:
\begin{enumerate}

\item (Generating the restriction set) Let $ V =\emptyset$. For every $i\in [t]$, if the number of flipped bits in $ S_i $ is less than  $19k_i$, then $V = V \cup S_i$ otherwise $V = V \cup \{j\mid \mbox{the } j \mbox{-th bit is flipped previously by this algorithm}\}$; (If a bit is flipped twice, then it is regarded as not flipped)

\item Find $j \in V$ s.t. the number of unsatisfied parity checks in $\Gamma_1(j)$ is larger than $|\Gamma_1(j)|/2 = d_1/2$; Flip the $j$-th bit, and restart this stage; If no such $j$, go to the next step;

\item Go to the next stage. 
\end{enumerate}

Stage 2 (classic belief propagation using $z_2$):
\begin{enumerate}

\item Apply the decoding of $C_0$ on the current $y$ concatenated with $z_2$.
 
\item Output the decoded message.
 
\end{enumerate}

\end{construction}

%For   step 1 of stage 1, for every $i\in [l]$, if the number of flipped bits in $ S_i $ is at least $19k_i$, we say that the set $S_i$  is on the flipping limitation.

\begin{lemma}
\label{lem:enhancedBPDec}

If $\distance(y_{\overline{S}}, x_{\overline{S}}) = 0, \forall i\in [l], \distance(y_{S_i}, x_{S_i}) \leq k_i $, 
then the decoder outputs $x$ correctly.
%$O( n\poly \log n )$. 

\end{lemma}

\begin{proof} 

\begin{claim}
The first stage ends in at most $O(m_1)$ rounds, and the number of errors in $y$ is reduced to be less than $2k' $.
\end{claim}

\begin{proof}

Let $A_\tau$ be the set of indices of tampered bits (comparing to $x$) in $y$ at (immediately before) the $\tau$-th round. At the beginning $|A_1| = \distance(y, x)$. 

We first show that if $ |A_\tau| \geq 2k'  $, then we can indeed find an index $j\in V$ s.t.  the number of unsatisfied parity checks in $\Gamma_1(j)$ is larger than $|\Gamma_1(j)|/2$. 

Denote
$A'_\tau = A_\tau \cap V$. Let $s, s'$ be the numbers of satisfied checks in $ \Gamma_1(A_\tau), \Gamma_1(A'_\tau)$. Let $ u, u' $ be the numbers of unsatisfied checks in   $ \Gamma_1(A_\tau), \Gamma_1(A'_\tau)$.

Consider $i\in [l]$ s.t. the number of flipped bits  is exactly $19k_i$. As $\distance(y_{S_i}, x_{S_i}) \leq k_i$,  the number of tampered bits in $S_i$ is at most $20k_i$. So $|A_\tau \cap S_i | \leq 20 k_i$,
  since $\distance(y_{S_i}, x_{S_i}) \leq k_i$. Also note that these tempered bits (at the beginning of the stage) can be flipped by the algorithm, we know the current number of tampered bits in $A'_\tau \cap S_i$ is at least $18k_i$. So 
\begin{equation}
|A'_\tau \cap S_i| \geq 0.9 |A_\tau \cap S_i|.
\end{equation} 

For $i \in [l]$ s.t. the number of flipped bits is less than $19 k_i$, since $V\cap S_i = S_i$,
\begin{equation}
|A'_\tau \cap S_i|  = |A_\tau \cap S_i|
\end{equation} 
As a result, noting that $S_i, i\in [l]$ are disjoint,
\begin{equation}
\label{pfenhancedBPDecieq4}
\frac{|A'_\tau|}{|A_\tau|} = \frac{\sum_i |A'_\tau\cap S_i| }{ \sum_i |A_\tau \cap S_i| } \geq 0.9, 
\end{equation}

As $ |A_\tau| \geq 2k' $, it holds $ |A'_\tau| \geq 1.8k' \geq k'$.
By the expansion property of $\Gamma_1$,
\begin{equation}\label{pfenhancedBPDecieq1}
s'+u' = |\Gamma_1(A'_\tau)| \geq 0.9 d |A'_\tau|.
\end{equation}

On the other hand, note that $2s +  u \leq d|A_\tau|$, since each satisfied check in $\Gamma_1(A_\tau)$ must have at least two bits in $x_{A_\tau}$ to be as addends.
As $A'_\tau =  A_\tau \cap V$, we have $s'\leq s, u'\leq u$. 

Thus
\begin{equation}\label{pfenhancedBPDecieq2}
2s'+u' \leq 2s+u \leq d|A_\tau|.
\end{equation}
Combining  (\ref{pfenhancedBPDecieq1}) and (\ref{pfenhancedBPDecieq2}), we get 
\begin{equation}
\label{pfenhancedBPDecieq3}
u'  \geq 2( 0.9d|A'_\tau| - 0.5d|A_\tau|).
\end{equation}

Further by  (\ref{pfenhancedBPDecieq4}),(\ref{pfenhancedBPDecieq3}),
\begin{equation}
\label{pfenhancedBPDecieq5}
u' \geq 0.68 d|A_\tau| \geq 0.68 d|A'_\tau|.
\end{equation}
Hence by an averaging argument, there is an index $j\in V$ s.t. the number of unsatisfied parity checks in $\Gamma_1 (j)$ is at least $ 0.68 d$.

As a result, after doing the flipping for this round, the number of unsatisfied parity checks is strictly decreased. Also note that because of the restriction sets in our algorithm our operation cannot create an $A_\tau$ in some steps s.t. it does not have a good expansion. Hence, the first stage ends when $ |A_\tau| < 2k' $.

Next we consider $|A_\tau| < 2k'$ at the beginning of a round $\tau$. There are two possible cases. 

The first case is that in step 2 the algorithm does not find a $j\in V$ to conduct the operation, so it will go to the next stage as desired. 

The second case is that
there is still an index $j\in V$ s.t.  the number of unsatisfied parity checks in $\Gamma_1 (j)$ is more than half. Hence after flipping, and the number of unsatisfied parity checks is again strictly decreased.
Note that there are at most $O(m_1)$ unsatisfied checks. So this procedure will end in at most $O(m_1)$ rounds. 

For either case, stage 1 will end with $|A_\tau| < 2k'$.
This shows the claim.

\end{proof}

As a result, after stage 1, the number of errors  is less than $2k' $. 

As $C_0$ can correct $2k'$ errors, by \cref{agcode},  the decoding algorithm outputs $x$ correctly.

\end{proof}

\begin{theorem}
\label{thm:DEHspecial}
There is an efficient one-way protocol   for every $(\mathbf{s}, \mathbf{k}, t)$ DE, arbitrary $ s_i =  k2^{\Theta(i)} , l = \Omega(\log \frac{n}{k}), k_i = \max\{k/2^{\Theta(i)}, \Theta(\frac{k}{ l \log  \frac{n}{k}  }) \} \leq s_i/40 $, $t  \leq O(\sqrt{l}  )$ , having communication complexity $O(k)$, success probability $1-2^{-\Theta(k \frac{\log \log \frac{n}{k}} {\log \frac{n}{k}  }) } $.

\end{theorem}

\begin{proof}

The protocol is constructed by \cref{constr:enhancedBPDec} and we will use a random $(n, m, d)$ bipartite graph to be $\Gamma_1$.
By \cref{lem:expanderGeneration}, a random bipartite $(n, m, d)$ graph $\Gamma_1$  is an $(n, m, d, \mathbf{S}, \mathbf{\overline{k}} , [k', k], 0.9  )$ expander, with failure probability at most $\varepsilon = 2^{-\Theta(k' \log \frac{2k}{k'})}$, where we let $ m = O(2dk), d = O(1) $, $ \overline{k}_i = 20 k_i,   i\in [t]$, $ k' = O(k/\log \frac{n}{k})$. 
Also  since $t  \leq O( \sqrt{l} ) $, we have $k'' t \log \frac{2k 2^t}{k''}\leq k' \log \frac{2k}{k'}$, where $k'' = O(\frac{k}{l\log \frac{n}{k}})$.

By   \cref{lem:enhancedBPDec}, Bob can compute $x$, by using $y,  z, \mathbf{S}, \mathbf{k}, k' $ and the common randomness.

The communication complexity is $|z| = m = O(k)$. The protocol is efficient since both encoding and decoding are efficient. The failure probability is $\varepsilon = 2^{-\Theta(k \frac{\log \log \frac{n}{k}} {\log \frac{n}{k}  }) }  $ since the construction of $\Gamma_1$ is the only part we use randomness.

\end{proof}

Note that \cref{thm:DEHspecialintro} directly follows from \cref{thm:DEHspecial} by letting $l = O(t^2)$.

\section{Optimal Document Exchange under Edit Distance}
\label{sec:DE}
In this section we give the one-way document exchange protocol for edit distance. We begin with a randomized protocol where the two parties  have shared randomness.

\begin{construction}
\label{Constr:OWDocExc}

The input string for Alice has length $n\in \mathbb{N}$ and there are totally   
$k \in [\Theta(  \log^4\frac{n}{k})  , \Theta( n)]$ 
edit errors between Alice's string and Bob's string. 

Both Alice's and Bob's algorithms have $L = O(\log \frac{n}{k})$ levels.
For every $i\in [L]$, in the $i$-th level,
\begin{itemize}
\item Let block size $b_i =  \frac{n}{3\cdot 2^{i} k } $, i.e., in each level we divide a block in the previous level evenly into two blocks; (We choose $L$ properly s.t. $b_L = O(\log \frac{n}{k})$)

\item The number of blocks $l_i = n/b_i$;
\end{itemize}

Alice: On input $x \in \{0,1\}^n$;
\begin{enumerate}[label*=\arabic*.]

\item For the $i$-th level,
\begin{enumerate}[label*=\arabic*.]

\item Partition $x$ into consecutive blocks $x[1, b_i], x[1+b_i, 2b_i], \ldots, x[1+(l-1)b_i, l_ib_i] $;

\item{\label{randhash}} Let $h_j: \{0, 1\}^{b_i} \rightarrow \{0,1\}^{c}, j \in [l_i]$ be a sequence of random hash functions with $c$ being a large enough constant positive integer; 

\item Compute $v[i][j] = h_j\left(x[1+(j-1)b_i, jb_i] \right), j\in [l_i]$; 

\item $v[i] = (v[i][1], \ldots, v[i][l_i])$;

%\item{\label{expandergraph}} Let $\Gamma: [n_1] \times [d_1] \rightarrow [m_1]$ be a $(10k, 0.9 d_1, \eps_1)$ random expander generated by Lemma \ref{lem:expanderGeneration}, with $\eps_1 = 1/\poly(k)$, $d_1  = O(1)$ large enough , $m_1 = O(k)$.

\item{\label{expanderencstep}} 
By the sketch construction of \cref{thm:DEHspecial}, compute $ z[i] \in \{0,1\}^{m = O(k)} $, a sketch of $v[i]$, the   expander  constructed in this step being $\Gamma : \{0,1\}^{l_i} \times \{0,1\}^{d_1 = 10} \rightarrow \{0, 1\}^{m }$;

%Let $\Gamma: [n_1] \times [d_1] \rightarrow [m_1]$ be a $(10k, 0.9 d_1, \eps_1)$ random expander generated by Lemma \ref{lem:expanderGeneration}, with $\eps_1 = 1/\poly(k)$, $d_1  = O(1)$ large enough , $m_1 = O(k)$.

%
%\item For every $p\in[c]$, compute the redundancy $ z[i][p] \in \{0,1\}^{m_1} $ for the $p$-th bits of $v[i][1], \ldots, v[i][l_i]$ by Construction \ref{expandercode};
%
%\item $z[i] = (z[i][1],\ldots, z[i][c])$;

\end{enumerate}

\item Compute the redundancy $z_{\mathsf{final}} \in (\{0,1\}^{b_L})^{\Theta(k)} $ for the blocks of the $L$-th level  by Theorem \ref{agcode}, where the code has distance $16k$;

\item Send   $z = (z[1], z[2], \ldots, z[L])$, $v[1]$, $z_{\mathsf{final}}$.

\end{enumerate}

Bob: On input $y \in \{0,1\}^{O(n)}$ and received $z$, $v[1]$, $z_{\mathsf{final}}$;

\begin{enumerate}[label*=\arabic*.]

\item Create $\tilde{x} \in \{0, 1, *\}^{n}$ (i.e. his current version of Alice's $x$), initiating it to be $(*, *, \ldots, *)$;

\item Let $ A_1 = [l_1],  A_i = \emptyset, i = 2, 3, \ldots, L$;

\item For the $i$-th level, where $1 \leq i \leq L-1$,
\begin{enumerate}[label*=\arabic*.]

\item Divide $\tilde{x}$ into length $b_i$ consecutive blocks, $\tilde{x}[1, b_i], \ldots, \tilde{x}[1+(l_i-1)b_i, l_i b_i]$;

\item Utilize the common randomness to get functions $h_j: \{0,1\}^{b_i} \rightarrow \{0,1\}^c, j\in [l_i]$ that Alice gets in her stage \ref{randhash}

\item Compute $ \tilde{v}[i] = \left( h_1(\tilde{x}[1, b_i]), \ldots, h_{l_i}(\tilde{x}[1+(l_i-1)b_i, l_i b_i]) \right) $;

%\item Generate $\Gamma: [n_1] \times [d_1] \rightarrow [m_1]$, the same graph as that in Alice's stage \ref{expandergraph};

\item For every $i' \in [\ell]$, let $S_{i'} \subseteq [l_i]$ be the indices of the (descendent) blocks in the current level, whose ancestors are those blocks indicated by $A_{i'}$, i.e.  $j$ is in $S_{i'}$ iff there is $j'  \in A_{i'}$ s.t. $ [1+(j-1)b_i, jb_i] \subseteq [1+(j'-1)b_{i'} , j'b_{i'} ] $;

%\item For every $p\in [c]$, compute the $p$-th bits of $(v[i][1], \ldots, v[i][l_i])$ by using the decoding algorithm from Construction \ref{constr:enhancedBPDec} on input $\Gamma$, the $p$-th bits of $( \tilde{v}[i][1], \ldots,  \tilde{v}[i][l_i])$, together with $S_q, k_q = , q\in [\ell]$;

\item Compute $v[i]$ by using the decoding algorithm from Construction \ref{constr:enhancedBPDec} on input $\tilde{v}[i]$, $S_{i'}$, $k_{i'} = \max(k/2^{0.9c(i-i')}, k/\log^3 \frac{n}{k} )$, $i' = i-1, i-2, \ldots, 1$, and the received $z[i]$;

\item Let $T_i = \emptyset$. For every $j \in [l_i] $, if $ v[i][j] \neq \tilde{v}[i][j]$, then put $j \in T_i$ and then check every $i' = i-1, i-2, \ldots, 1$, if the $ j$-th block in the current level is a descendent of the $j'$-th block in the $i'$-th level, then remove $j'$ from  $A_{i'}$;

\item Let $A_i = T_i$;

%\item Compute $w = ((\rho'_1, \rho_1), \ldots, (\rho'_{|w|}, \rho_{|w|})) \in ( A_i \times [|y|])^{|w|}$ which is the maximum monotone matching between $x$'s blocks indicated by $A_i$, and $y$, under $h_1, \ldots, h_{l_i}$, using $v[i]$, by Lemma \ref{dpformatch};
%
%\item{\label{refill}} Evaluate $\tilde{x}$ according to $w$, i.e. let  $\tilde{x}[\rho'_j] = y[\rho_j, \rho_j + b_i - 1]$;

\item Compute $w  \in ( A_i \times [|y|])^{|w|}$ which is the maximum monotone matching between $x$'s blocks indicated by $A_i$, and $y$, under $h_1, \ldots, h_{l_i}$, using $v[i]$, by Lemma \ref{dpformatch}; (We interpret $w$ as a sequence of matches, the $j$th match being denoted as  $(w[j][1], w[j][2])$.)

\item{\label{refill}} Evaluate $\tilde{x}$ according to $w$, i.e. let  $\tilde{x}[w[j][1]] = y[w[j][2], w[j][2] + b_i - 1], \forall j\in [|w|]$;

\end{enumerate}

\item In the $L$'th level, apply the decoding of Theorem \ref{agcode} on the blocks of $\tilde{x}$ and $ z_{\mathsf{final}}$ to get $x$;

\item Return $x$.

\end{enumerate}

\end{construction}

Next we show the correctness of our construction.
%
%
%\begin{lemma}
%
%For every $i\in [L]$, the maximum monotone matching between $ x$'s blocks $x[1+ (j-1)b_i, jb_i], j\in [T_i]$, and $y$, under $h_j, j\in T_i$ has size at
%least $|S|-k$.
%
%\end{lemma}
%
%\begin{proof}
%
%Note that there are totally $k$ insertions/deletions transforming $x$ to $y$. These operations can happen in at most $k$ blocks among $x[1, b_i], \ldots, x[1+(l_i-1)b_i, l_i]$. So at least $|T_i|-k$ of them will be in $y$. So the maximum monotone matching between $x$'s blocks and $y$ under $h_j, j\in T_i$ has length at least $|T_i|-k$.
%
%
%\end{proof}

Consider every level $i\in [L]$, every $i' = i-1, i-2, \ldots, 1$. We denote the  set descendants in the $i$-th level, stemming from $A_{i'}  $,  as $\tilde{A}_{i'}$.  The indices  set of undetected wrongly recovered blocks  in $\tilde{A}_{i'}$, is denoted as $B_{i'}$, $i' = i-1, \ldots, 1$.

Let $ i^* $ be s.t. $ k''  \triangleq k'/\Theta(\log^2 \frac{n}{k}) \in [k/2^{c(i-i^*)}, k/2^{c(i-i^*+1)}] $, $k'\triangleq  k /\log \frac{n}{k} $.

\begin{lemma}
\label{lem:veryOldAncestors}

For every $i\in [L]$, if $\forall i'<i, |T_{i'}| = O(k)$, and $v[i']$ are computed correctly by Bob,
then 
\begin{itemize}
\item for every $i' \leq i^*$,   the probability that $|B_{i'}|\geq k''$ is at most $ 2^{-\Omega(k'' )} $;

\item for every $i' \in ( i^*, i)$, the probability that $|B_{i'}| \geq k_{i'} =  k/2^{0.9c (i-i')}$ is at most $2^{-\Omega( ck/2^{c(i-i')} )  }$.
\end{itemize}

\end{lemma}

\begin{proof}

Consider the possibilities of $B_{i'}$. 
Each possibility can be described by a $w$-witness with $w = |B_{i'}|$. The witness is a sequence of (number) $w$ indices where each index is  in the $i$-th level indicating a wrongly recovered block. This sequence is further partitioned into $ i-i'+1 $ groups corresponding to levels $i', i'+1, \ldots, i$.  We numerate these groups as group $i', i'+1, \ldots, i$. 

Consider the trees rooted at blocks in $A_{i'}$. Each of them has height $ i-i' $. Each node is a block in a certain level between $i'$ and $i$. 

The $w$-witness describes level $i$ bad blocks which are descendants of blocks in $A_{i'}$, uniquely in the following way.

Group $j \in [i', i]$ consists of indices of bad blocks, one for each depth $i-j$ tree whose root is a wrong block in level $j$.  Note that for one tree, there may be many bad leaf blocks. For this case, we only pick the leftmost wrong one. These forms the group one. After each picking, we cut all the edges from that block to the root. This gives $i-i'+1$ sub-trees. One of them is the last block. We only focus on sub-trees other than that picked block. They have depth from $1$ to $i-j$. We update the set of trees by adding these trees from cutting and delete the trees being cut. 

In this way, every error patterns can be described.  This is because, every leaf node is either being picked or still in one of the trees in the forest. Once the leaf is in one of the trees in the forest, it can be picked in a certain level of the picking procedure.

Let the number of wrong blocks being picked for each level $j$ be $ w_{ j} $. 
%Note that it is at most $O(k)$ by the assumption. 

The total number of error patterns is
\begin{equation}
\begin{split}
P = & {k \choose w_{ i'}}  \cdot 2^{ (i-i')w_{ i'} }  \cdot  { w_{ i'} \choose w_{ i'+1} }  \cdot   2^{ (i-i'-1)w_{ i'+1} }   \cdot    {{ w_{ i'}+ w_{ i'+1}}   \choose w_{ i'+2}  } \cdot 2^{ (i-i'-2)w_{ i'+2} }   \cdots { {\sum_{j=i'}^{i-1} w_{ j} } \choose w_{ i}  }   \\
\leq &  {k \choose  w_{ i'}} { {  \sum_{j=i'}^{i-1} (i-j) w_{ j}  }  \choose {\sum_{j=i'}^{i }   w_{ j} } } \cdot 2^{  \sum_{j=i'}^{i-1} (i-j) w_{ j}}  \\  
\end{split}
\end{equation}

For $i' \leq i^*$,
suppose $\sum_{j = i'}^{i } (i-j) w_{ j} =  k''$.
Then
\begin{equation}
\begin{split}
P\leq &  {k \choose k''/(i-i')}  \cdot 2^{2k_0}\\
\leq & 2^{O(k'') \cdot \frac{O(\log k)}{i-i'}} \cdot 2^{2 k''}\\
\leq & 2^{O(k'')}.
\end{split}
\end{equation}  
 
Note that the probability that a specific error pattern happens is at most $ 2^{- c\sum_{j = i'}^{i} (i-j) w_{ j} } = 2^{- ck''}$ because each block in group $j$ is checked for $i-j$ times independently. Since $c$ is a large enough constant, $\sum_{j = i'}^{i} (i-j) w_{ j}$ is a integer in $ [0, \poly(k \log n)] $, we know by a union bound, $\sum_{j = i'}^{i} (i-j) w_{ j} \geq  k''$ happens with probability at most $ 2^{-ck''} \times 2^{ O( k'')} \times \poly( k \log n) \leq 2^{-\Omega(k'')}$.

%As a result, by a union bound $w =\sum_{j = i'}^{i}   w_{ j} \geq k''$ happens with probability at most $2^{-\Omega(k'')}$ . 

For $i > i^*$, suppose $\sum_{j = i'}^{i} (i-j) w_{ j} =   k/2^{0.9c (i-i')} =  k_{i'}$. Then 
\begin{equation}
\begin{split}
P\leq &  {k \choose  k_{i'}/(i-i')}  \cdot 2^{2  k_{i'}}\\
\leq & 2^{ (0.9c (i-i') +O(1)+\log (i-i')  ) \cdot  k_{i'}/(i-i')    } \cdot 2^{2  k_{i'}}\\
\leq & 2^{  0.91ck_{i'}},
\end{split}
\end{equation}  
when $c$ is a large enough constant.

Similarly, note that the probability that a specific error pattern happens is at most $ 2^{- c\sum_{j = i'}^{i} (i-j) w_{ j} } = 2^{- c k_{i'} }$ because each block in group $j$ is checked for $i-j$ times independently. Since $c$ is a large enough constant, $\sum_{j = i'}^{i} (i-j) w_{ j}$ is a integer in $ [0, \poly(n)] $, we know by a union bound, $\sum_{j = i'}^{i} (i-j) w_{ j} \geq  k_{i'}$ happens with probability at most $ 2^{-ck_{i'}} \times 2^{  0.91ck_{i'}} \times \poly( k \log n) \leq 2^{-\Omega( k_{i'} ) }$.

As a result, $w =\sum_{j = i'}^{i}   w_{ j}> k_{i'}$ happens with probability at most $2^{-\Omega( k_{i'})}  \leq 2^{-\Omega(k'')}$ . 

\end{proof}

\begin{lemma}
\label{Badsetsnotlarge}
For every $i\in [L]$, if $\forall i'<i, |T_{i'}| = O(k)$, and $v[i']$ are computed correctly by Bob, then with probability at least $1-2^{-\Omega(k'')}$,
$$ \sum_{i' = 1}^{i-1}  |B_{i'}| < k.  $$

\end{lemma}

\begin{proof}

By Lemma \ref{lem:veryOldAncestors}, 
for $i' \leq i^*$, with probability at least $1 - 2^{-\Omega(k'')}$, $|B_{i'}| < k''$; for $i'> i^*$, with probability at least $1 - 2^{-\Omega(k'')}$, $ |B_{i'}| \leq k_{i'} =  k/2^{0.9c (i-i')} $.

By a union bound, with probability at least $ 1 - i2^{-\Omega(k'')}  = 1 - 2^{-\Omega(k'')} $,
$$ \sum_{i' = 1}^{i-1} |B_{i'}| = \sum_{i'=1}^{i^*} |B_{i'}| + \sum_{i' = i^*+1}^{i-1} |B_{i'}|   \leq  (i^*-1)k'' + 0.5 k < k. $$

\end{proof}

\begin{lemma}
\label{lem:numOfWronglyRecovered}
For every $i\in L$, at level $i$, if $v[i]$ are computed correctly by Bob, and  $|T_i| \leq 6k$, then  with probability $1- 2^{-\Theta(k)}$,  the number of wrongly recovered blocks introduced by $w$ is at most $k$.

\end{lemma}

\begin{proof}

Assume the number of wrongly recovered blocks introduced by $w$ is larger than $k$. Then there more than $k$ pairs in the matching are bad pairs. This happens with probability $1/2^{c k}$. 

Note that by Lemma \ref{dpformatch}, for $w$, \[  |\rho'_1 - \rho_1| + |(\rho'_2 - \rho'_1 ) - (\rho_2 -\rho_1) | + \cdots + | ( \rho'_{|w|} -\rho'_{|w|-1}) - (\rho_{|w|} -\rho_{|w|-1})| \leq  k. \]

By Lemma \ref{numOfPossibleMatchings}, since $|T_i| \leq 6k$, there are totally $2^{O(k)}$ possible matchings that can be output by our algorithm. 

So by a union bound, the conclusion holds with probability $1-2^{-\Theta(k)}$.

\end{proof}

\begin{lemma}
\label{lem:numOfWrBlksAfterRefill}
For every $i\in L$, in level $i$, if $v[i]$ are computed correctly, and  $|T_i| = O(k)$, then  with probability $1-2^{-\Theta(k)}$, the number of wrongly recovered blocks and uncovered blocks in $T_i$ after \ref{refill} is at most $ 2k $. 

\end{lemma}

\begin{proof}

 By Lemma \ref{dpformatch}, $|w| \geq |T_i| - k$. Thus the number of uncovered blocks is at most $k$. By  Lemma \ref{lem:numOfWronglyRecovered}, with probability $1-1/2^{\Theta(k)}$, the number of wrongly recovered blocks introduced by $w$ is at most $k$. So the total number of wrongly recovered blocks is at most $2k$.

\end{proof}

\begin{lemma}
\label{lem:OWmaininduction}
For every $i\in L$, with probability $1-2^{-\Theta(k'')}$, 
\begin{itemize}
\item after the first step of level $i$, the number of wrongly recovered blocks is at most $ 6k$;

\item Bob can compute $v[i]$ correctly;

\item the number of wrongly recovered blocks in $T_i$ is at most $2k$ after step \ref{refill}.

\end{itemize}

\end{lemma}

\begin{proof}

We use induction. 

In the first level, $\tilde{x} = (*, *, \ldots, *)$. So the number of wrongly recovered blocks at the beginning is $l_1 = n/b_1 = 6k$. So The number of wrongly recovered blocks is at most $ 6k$. Also Bob can get $v[1]$ correctly, since it is directly sent by Alice. By Lemma \ref{lem:numOfWrBlksAfterRefill},   with probability $1-1/2^{\Theta(k)}$,   the total number of wrongly recovered blocks is at most $2k$ if we regard uncovered blocks as wrongly recovered.

Suppose the conclusion holds for the first $i-1$-level. Consider level $i$. 

%By Lemma \ref{lem: veryOldAncestors} with probability $1-1/\poly(k)$, there no wrongly recovered blocks in   $A_{i-\ell -1}, \ldots, A_1$. For  $i' = i-1, i-2, \ldots, i-\ell$, $|A_{i'}| \leq T_{i'} $. So $|S_{q}| \leq |A_{i-q}| 2^{q} \leq 6k 2^q , q = 1,\ldots, \ell$. 
By Lemma \ref{Badsetsnotlarge}, with probability $1-2^{-\Omega(k'')}$,  the total number of wrongly recovered blocks is $\sum_{i'=1}^{i-1}|B_{i'}|  <k$.  

%In Bob's $(i-q)$-th round, the number of wrongly recovered blocks in $T_i$ is at most $2k$ by assumption. Each such block is randomly hashed and checked for $q$ times. So the expectation of wrongly recovered but undetected blocks is $ \frac{2k \times 2^{q}}{2^{c q}} \leq \frac{k}{2^{(c-1)q -1 }}  $. We know $k_q = \frac{k}{2 \times 10^q}$. We pick $\ell$ s.t. $10^{\ell} = \Theta(k^{0.1})$. We also let $c$ be large enough constant s.t. $ 2^{(c-1)q-1} = \omega( 10^q ) $. By Chernoff bound, 
%\[ \Pr[\mbox{the number of undetected wrong blocks in } A_q \mbox{ is at most } k_q ] \geq 1- 2^{- \Theta(k_q)}  \geq 1- 2^{-\Theta(k^{0.9})}. \] 

%By a union bound, with probability at least $1-  \ell2^{-\Theta(k^{0.9})} \geq 1-  O(\log k)2^{-\Theta(k^{0.9})}$, for every $q \in \ell$, the number of undetected wrong blocks in $A_q$ is at most $k_q$.  This means the total number of  wrong blocks in level $i$ is at most $ \sum_{q  = 1}^{\ell} k_q 2^q \leq 2k$. So The number of wrongly recovered blocks is at most $  6k$.

By Lemma \ref{lem:expanderGeneration}, with probability $1- \eps_1  = 1-2^{-\Omega(k') }$, $\Gamma_1$ is
 a bipartite graph, having $n_1= l_i$ left vertices, $m =  O(k)  $ right vertices,  left degree $d = O(1)$,  s.t. $\forall R \subseteq [n_1], |R|   \in [k', k],  |R\cap S_{i'}|  \leq k'_{i'} = \max(20k/{2^{c(i-i')}},  20k^{0.9} )  $, 
 $$  \Gamma(R) > 0.9 d|R|.$$ 
Note that $k'_{i'}  \geq 20 k_{i'}$.
Also note that $ i' $ iterates in $[1, i-1]$. So the number of $S_{i'}$ is at most $L \leq k^{\beta/2}\sqrt{\log k}  $.
So by \cref{thm:DEHspecial}, Bob can get the correct $v[i]$.

As a result, by a union bound with probability $1-L2^{-\Theta(k'')}$, Bob can compute $v[i]$ correctly. Note that $L  =O(\log \frac{n}{k})$, $k = \Omega( \log^4 \frac{n}{k} )   $. So the probability is at least $1-2^{-\Theta(k'')}$.

By Lemma \ref{lem:numOfWrBlksAfterRefill},   with probability $1-1/2^{\Theta(k)}$,   the total number of wrongly recovered blocks in $T_i$ is at most $2k$ after stage \ref{refill}.

So the overall   probability is as desired.

This shows the inductive step.

\end{proof}

\begin{lemma}
\label{lem:BobComputesX}

With probability $1-2^{-\Theta(k'')}$, Bob outputs $x$ correctly.

\end{lemma}

\begin{proof}

By Lemma \ref{lem:OWmaininduction}, with probability $1-2^{-\Theta(k'')}$, at the last level, there are at most $6k$ wrong blocks. Since $z_{\final}$ is the redundancy for a code with distance $16k$, all wrong blocks can be corrected. So Bob computes $x$ correctly.

%Note that by a union bound  the success probability is $  1-L2^{-\Theta(k^{0.9})} = 1-2^{-\Theta(k^{0.9})}$, because $k =  , L = O(\log \frac{n}{k})$, 
\end{proof}

\begin{lemma}
\label{lem:OWCC}
The communication complexity is $O(k \log \frac{n}{k})$.

\end{lemma}

\begin{proof}
Note that since $m = O(k)$, $z[i] = O(k)$. Also note that $|v[1]|  = O(k)$, as the output length for of the hash function is $O(1)$ and $l_1 = O(k)$. $|z_{\final} | = O(k\log \frac{n}{k})$ by Theorem \ref{agcode}.

So the overall communication complexity is $ \sum_{i=1}^{L} |z[i]| + |v[1]| + |z_{\final}| = O(k \log \frac{n}{k}) $.

\end{proof}

\begin{theorem} \label{thm:crOWProtocol}
There exists an efficient one-way edit distance document exchange protocol using common randomness, for every  $n\in \mathbb{N}$,  $k = \Omega( \log^4 \frac{n}{k}) $, having sketch length $O(k \log \frac{n}{k})$,   success  probability $1 -   2^{-\Omega(k/\log^3 \frac{n}{k})}$.

%The encoding   time is $\tilde{O}(n^2)$. The decoding  time is $\tilde{O}( \min(n^3, k^2 n^2) )$.
\end{theorem}

\begin{proof}
It immediately follows from Lemma \ref{lem:BobComputesX}, 
%\ref{lem:OWTimeC}, 
\ref{lem:OWCC}.
The protocol is efficient since all components and steps are efficient.

\end{proof}

By combining \cref{thm:crOWProtocol} and the result of Haeupler \cite{haeupler2018optimal}, we immediately get the following.

\begin{theorem} \label{thm:crOWProtocol1}
There exists an efficient one-way edit distance document exchange protocol using common randomness, for every $n, k \in \mathbb{N}$, having sketch length $O(k \log \frac{n}{k})$,   success  probability $1 -   \min\{2^{-\Theta(k/  \log^3 \frac{n}{k} ) }, 1/\poly(n) \} $.

%The encoding   time is $\tilde{O}(n^2)$. The decoding  time is $\tilde{O}( \min(n^3, k^2 n^2) )$.
\end{theorem}

\begin{proof}

When $k = \Omega(\log^4 \frac{n}{k})$, we use \cref{thm:crOWProtocol}. Otherwise we use the random protocol from \cite{haeupler2018optimal} which has success probability $1-1/\poly(n)$. Both of them have the sketch length as desired.

\end{proof}

\subsection{ Removing Shared Randomness}

In Construction \ref{Constr:OWDocExc}, we use   common randomness to generate hash functions $h_j, j\in [l_i]$ for each $i\in[L]$. Also we use common randomness to generate the random bipartite graph $\Gamma$ for the encoding of the hash values. Now we show that we can use almost $ \kappa$-wise independence generator to reduce randomness.

\begin{lemma}
\label{lem:wocommonrand}
%For every $i\in [L]$,
%let $g:\{0,1\}^{d } \rightarrow ((\{0,1\}^{c})^{2^{b_i}})^{l_i } $ be an $2^{-10ck}$-almost $5ck$-wise independence generator, where $d = O(k + \log (k (b_i + \log l_i))) = O(k + \log (kn))$.
%Let $h_j(u ) = g( r )[j][u]$, where $r$ is a random seed.

Replace the common randomness used in Construction \ref{Constr:OWDocExc},
\begin{itemize}

\item for generating hash functions,   by an $\epsilon$-almost $ 10c k $-wise independent distribution, with $\epsilon = 2^{-10c k}$;

\item for generating $\Gamma_1$,    by $O(k)$-wise independent distributions over alphabet $[m]$.
(Recall that $m = O(k)$)
Then with probability $1-2^{-\Theta(k')}$, Bob outputs $x$ correctly.

\end{itemize} 

\end{lemma}

\begin{proof}

We need to recompute the following probabilities.

In Lemma \ref{lem:veryOldAncestors}, a specific error pattern happens with probability at most $ 2^{-c k''} \pm \epsilon \leq 2^{- 0.9ck''}$.

In Lemma \ref{lem:numOfWronglyRecovered},
if there are $k$ wrongly matched blocks introduced by $w$, then there are $ k $ hash collisions each for a different $h_j$ in level $i$. So the probability is at most $2^{-2ck}\pm 2^{-10ck}= 2^{-\Theta(k)}$.

The rest of the analysis of the above two lemmas can still go through. These two lemmas are the only two in the proof of Lemma \ref{lem:BobComputesX} which will use the independence of hash functions.

As a result, the proof of Lemma \ref{lem:BobComputesX} can still go through.
\end{proof}

\begin{theorem} \label{thm:OWProtocol}
There exists an efficient one-way edit distance document exchange protocol, for every   $k =\Omega( \log^4\frac{n}{k}  )  $, having sketch length $O(k \max\{ \log \frac{n}{k}, \log k \} )$,   success  probability $1 -   2^{-\Omega(k/\log^3 \frac{n}{k})}$.

%The encoding   time is $\tilde{O}(n^2)$. The decoding  time is $\tilde{O}( \min(n^3, k^2 n^2) )$.
\end{theorem}

\begin{proof}
Consider replacing the common randomness used in Construction \ref{Constr:OWDocExc} in the way of Lemma \ref{lem:wocommonrand}.
By Theorem \ref{kwiseg}, we can use a generator $g_1$ of seed length $O(k \max\{ \log k,  \log \frac{n}{k} \})$ to generate the $O(k)$-wise independent distribution. 
By Theorem \ref{almostkwiseg} we can use a generator $g_2$ of seed length $O(\log \frac{k\log n}{\epsilon})$ to generate the $\epsilon$-almost $10ck$-wise independent distribution. 

So we only need to let Alice send the seeds for these two, which have total length $O(k\log \frac{n}{k})$. Adding the communication complexity calculated by Lemma \ref{lem:OWCC}, the overall communication complexity is as desired.

%The overall construction is in polynomial time following from Lemma \ref{lem:OWTimeC}, and the fact that $g_1$ is explicit and  $g_2$ is highly explicit.

The correctness and success probability follows from \ref{lem:wocommonrand}. The protocol is efficient since all components and steps are efficient.

%
%
%For the encoding/decoding time, note that computing the $O(k)$-wise independence from Theorem \ref{kwiseg} will take $\tilde{O}(n^2)$ time. Computing a hash value  by using the  $\epsilon$-almost $ 10c k $-wise independence generator is at most $ \tilde{O}(  n+\log (1/\eps) ) = \tilde{O}( n)$. Computing all these hash values takes time at most $\tilde{O}(n^2)$
%
%Combining with the analysis from Lemma \ref{lem:OWTimeC}, the overall encoding time is $\tilde{O}( n^2)$.
%
%For the decoding time, similarly combining with the analysis from Lemma \ref{lem:OWTimeC}, it is $\tilde{O}(\min(n^3, k^2n^2))$.

\end{proof}

\section{Asymmetric Document Exchange with Two Sided Information}\label{sec:twoside}
In this section we study document exchange with two sided asymmetric information. We have the following definition.

\begin{definition}
There are two parties Alice and Bob. Alice has a string $x\in \{0,1\}^n$ and Bob has a string $y\in \{0,1\}^n$. Alice knows a vector of disjoint subsets $\svec^A=(S^A_1, \cdots, S^A_{t^A})$ and a vector of integers $\knum^A=(k^A_1, \cdots, k^A_{t^A})$. Bob knows a vector of disjoint subsets $\svec^B=(S^B_1, \cdots, S^B_{t^B})$ and a vector of integers $\knum^B=(k^B_1, \cdots, k^B_{t^B})$. That is, within each set $S^A_i$ or $S^B_i$, the Hamming distance between $x$ and $y$ is at most $k^A_i$ or $k^B_i$. Now one party tries to learn the string of the other party. 
\end{definition}

Again, let $\snum^A=(s^A_1, \cdots, s^A_t)$ where $\forall i, s^A_i=|S^A_i|$. Similarly, let $\snum^B=(s^B_1, \cdots, s^B_t)$ where $\forall i, s^B_i=|S^B_i|$. We call this problem an $(\snum^A, \snum^B, \knum^A, \knum^B, t^A, t^B)$ asymmetric document exchange (DE) problem, and we require the protocol to succeed for all possible configurations of the subsets $\svec^A=(S^A_1, \cdots, S^A_{t^A})$, $\svec^B=(S^B_1, \cdots, S^B_{t^B})$, and all possible strings $x, y$ that are consistent with the parameters.

We also have both lower bounds and upper bounds. 

\begin{theorem}
In an $(\snum^A, \snum^B, \knum^A, \knum^B,  t^A, t^B)$ asymmetric DE problem, suppose Bob learns Alice's string. Let $s^A=\sum_{i=1}^{t} s^A_i$ and $s^B=\sum_{i=1}^{t} s^B_i$, and assume $s^A+s^B \leq n$. Let $k^A=\sum_{i=1}^{t} k^A_i$ and $k^B=\sum_{i=1}^{t} k^B_i$. Then any deterministic protocol has communication complexity at least $\ent(n-s^B, k^A)+\ent(\snum^B, \knum^B)$, and any randomized protocol with success probability $\geq 1/2$ has communication complexity at least $\ent(n-s^B, k^A)+\ent(\snum^B, \knum^B)-1$. In addition, if $\forall i, s^B_i \geq 2 k^B_i$, then any one round deterministic protocol has communication complexity at least $\ent(n, k^A+k^B)$. This holds even if both parties know $(\snum^A, \snum^B)$ and $(\knum^A, \knum^B)$.
\end{theorem}

\begin{proof}
The proof is similar to the one sided case. For a deterministic protocol, assume for the sake of contradiction that there is a protocol with communication complexity less than $\ent(n-s^B, k^A)+\ent(\snum^B, \knum^B)$. Then fix Bob's string $y$ and there exist two different $x$'s that produce the same transcript, and in addition the inputs to Bob are the same. Thus Bob will not be able to distinguish the two $x$'s, a contradiction. The case of a randomized protocol is essentially the same up to an averaging argument.

For the case of one round deterministic protocol, again the argument is similar as before. Assume for the sake of contradiction that there is a protocol with communication complexity less than $\ent(n, k^A+k^B)$. Fix Bob's string $y$ and the number of different $x$'s within Hamming distance $k^A+k^B$ is exactly $2^{\ent(n, k^A+k^B)}$. For each such $x$, one can arrange the first at most $k^A$ differences to happen in $\svec^A$, and the rest of at most $k^B$ differences to happen in $\svec^B$, such that the subsets in $\svec^A$ and $\svec^B$ are all disjoint (since $s^A+s^B \leq n$). Note that each $x$ gives a vector $\svec^A$, and the one round transcript is a deterministic function of $(x, \svec^A, \snum^A, \snum^B, \knum^A, \knum^B)$, two different $x$'s will produce the same transcript. At this point, one can define a vector $\svec^B$ consistent with $\snum^B$, $\knum^B$ and both of the $x$'s (since $\forall i, s^B_i \geq 2 k^B_i$). This means the inputs to Bob are the same for the two $x$'s. Since Bob's final output is a deterministic function of the transcript and $(y, \svec^B, \snum^A, \snum^B, \knum^A, \knum^B)$, Bob will not be able to distinguish the two $x$'s, a contradiction.
\end{proof}

The positive result directly follows from the one-side result i.e. \cref{thm:GCAIchi}.

\begin{theorem}

There exists an explicit protocol for all $(\mathbf{s}^A, \mathbf{s}^B, \mathbf{k}^A, \mathbf{k}^B, t^A, t^B)$ 
DE , having communication complexity $O\left(  \left(\chi(\mathbf{s}^B, \mathbf{k}^B, t^B )+1 \right)^2  \left( \mathsf{H}(n-s^B, k^A) + \mathsf{H}(\mathbf{s}^B, \mathbf{k}^B) \right)   \right)$, success probability $1-2^{ -\Omega(   \min(k_t, k^A)   ) - 1/\poly(s^A+ s^B) }$, to let Bob learn Alice's string.

\end{theorem}

\begin{proof}

The two party can just think that there are at most $k^A$ errors   in set $[n] - S^B$. This contribute one more set (and its error bound) to the error pattern. And the problem becomes a one-side asymmetric information problem. So we can apply  \cref{thm:GCAIchi}  and the conclusion follows.

\end{proof}

\bibliographystyle{plain}
\bibliography{ref}

\end{document}